\documentclass[letterpaper,11pt]{article}

\usepackage[english]{babel}
\usepackage[utf8x]{inputenc}
\usepackage[T1]{fontenc}

\usepackage[margin=1in]{geometry}
\usepackage{todonotes}

\usepackage{amsthm,amsfonts,amssymb,amsmath,amsxtra}
\usepackage{algorithm,algpseudocode}
\usepackage{graphicx, xcolor}
\usepackage{lmodern,hyperref,paralist}
\usepackage{bookmark}
\usepackage{multicol,xspace}
\usepackage{tikz}
\usepackage{comment}
\usepackage{mathtools}

\usepackage{subcaption}
\usepackage{authblk}
\definecolor{DarkDesaturatedBlue}{HTML}{3A3556}
\definecolor{PureOrange}{HTML}{FFBA00}
\definecolor{VerySoftBlue}{HTML}{B5AFDB}

\newcommand{\OPT}{\operatorname{OPT}}
\newcommand{\opt}{\mathrm{opt}}
\newcommand{\alg}{\mathrm{alg}}

\newcommand{\cvca}{\textsc{\bf cycle VCA}\xspace}

\newcommand{\ALG}{\operatorname{ALG}}

\usepackage{etoolbox}

\newtheorem{theorem}{Theorem}
\newtheorem{lemma}[theorem]{Lemma}
\newtheorem{proposition}[theorem]{Proposition}

\newtheorem{definition}[theorem]{Definition}

\title{Approximation Algorithms for Vertex-Connectivity Augmentation on the Cycle\footnote{Full version of the extended abstract accepted at WAOA 2021. In that extended abstract an approximation ratio of $1.8703$ was claimed, but it has been corrected to $1.8704$ in this version. We apologize for the inconvenience.}}
\author[1]{Waldo G\'alvez\thanks{Email: \texttt{waldo.galvez@uoh.cl} This project was carried out when the author was a postdoctoral researcher at the Department of Computer Science of the Technical University of Munich in Germany. Supported by the European Research Council, Grant Agreement No. 691672, project APEG.}}
\author[2]{Francisco Sanhueza-Matamala\thanks{Email: \texttt{fsanhueza@dim.uchile.cl} Partially supported by grants ANID-PFCHA/Mag\'ister Nacional/2020- 22201780 and
FONDECYT Regular 1190043, and by ANID via Fondecyt 1181180 and PIA AFB170001.}}
\author[2,3]{Jos\'e A. Soto\thanks{Email: \texttt{jsoto@dim.uchile.cl} Partially supported by ANID via Fondecyt 1181180 and PIA AFB170001.}}

\affil[1]{Institute of Engineering Sciences, Universidad de O'Higgins, Chile}
\affil[2]{Departamento de Ingenier\'ia Matem\'atica, Universidad de Chile, Chile}
\affil[3]{Centro de Modelamiento Matem\'atico, IRL 2807 CNRS, Universidad de Chile, Chile}

\date{}

\begin{document}

\bibliographystyle{plainurl}

\makeatletter
\newcounter{phase}[algorithm]
\newlength{\phaserulewidth}
\newcommand{\setphaserulewidth}{\setlength{\phaserulewidth}}
\newcommand{\phase}[1]{%
  \vspace{-2.25ex}
  % Top phase rule
  \Statex\leavevmode\llap{\rule{\dimexpr\labelwidth+\labelsep}{\phaserulewidth}}\rule{\linewidth}{\phaserulewidth}
  \Statex\strut\refstepcounter{phase}\textit{Phase~\thephase~--~#1}% Phase text
  % Bottom phase rule
  \vspace{-2ex}\Statex\leavevmode\llap{\rule{\dimexpr\labelwidth+\labelsep}{\phaserulewidth}}\rule{\linewidth}{\phaserulewidth}}
\makeatother

%\AtEndEnvironment{proof}{\phantom{}\qed} %Automatically add qed symbol to the end of proof environment

\setphaserulewidth{.7pt}

%Para quitar numeracion del algoritmo
%\makeatletter
%\renewcommand{\fnum@algorithm}{\fname@algorithm}
%\makeatother

\maketitle              % typeset the header of the contribution
\thispagestyle{empty}
\begin{abstract}
Given a $k$-vertex-connected graph $G$ and a set $S$ of extra edges (links), the  goal of the $k$-vertex-connectivity augmentation problem is to find a set $S' \subseteq S$ of minimum size such that adding $S'$ to $G$ makes it $(k+1)$-vertex-connected. Unlike the edge-connectivity augmentation problem, research for the vertex-connectivity version has been sparse.

In this work we present the first polynomial time approximation algorithm that improves the known ratio of 2 for $2$-vertex-connectivity augmentation, for the case in which $G$ is a cycle.  This is the first step for attacking the more general problem of augmenting a $2$-connected graph.

Our algorithm is based on local search and attains an approximation ratio of $1.8704$. To derive it, we prove novel results on the structure of minimal solutions.

%\keywords{Approximation Algorithms \and Connectivity Augmentation \and Cycle augmentation \and Network Design}
\end{abstract}

\newpage
\setcounter{page}{1}
\section{Introduction}

In the field of \textbf{Survivable Network Design}, one of the main goals is to construct robust networks (e.g.~transportation, telecommunication or electric power supply networks), meaning that they are resilient to failures of nodes or edges (for instance due to attacks or malfunctioning), at a low cost. A classical example is the \textbf{Connectivity Augmentation} problem, where we are given a $k$-connected\footnote{For $k\in \mathbb{N}$, a $k$-connected graph is a graph $G=(V,E)$ satisfying that, for any $V'\subseteq V$ with $|V|\leq k-1$, $G$ remains connected after the deletion of $V'$. If the definition holds when replacing nodes by edges, the graph is said to be $k$-edge-connected.} ($k$-edge-connected, resp.) undirected graph $G=(V,E)$ and a collection $S$ of extra edges (\emph{links}), and the goal is to select the smallest possible set of links $S'\subseteq S$ so that $G'=(V,E\cup S')$ becomes $(k+1)$-connected ($(k+1)$-edge-connected, resp.).

Most of the variants of Connectivity Augmentation are known to be NP-hard. In the case of augmenting the edge-connectivity of a graph, a classical result of Khuller and Vishkin ~\cite{Khuller1994} provides already a $2$-approximation algorithm for this problem, and crucial breakthroughs have been developed in recent years (see~\cite{Grandoni20,Cecchetto20}). However, results for the case of vertex-connectivity augmentation are more scarce and the developed techniques are arguably more involved.

In this work, we focus on the case of augmenting the vertex-connectivity of a given $2$-connected graph by one. For this case, Auletta et al.~\cite{Auletta99} provide a $2$-approximation for the problem even when links have different weights and the goal is to minimize the total weight of the solution, but no improvement on the problem has been developed ever since.

\subsection{Our results}

In order to make progress on the latter problem, we study the case where the input $2$-connected graph is a cycle, which we denote as the \emph{cycle vertex-connectivity augmentation} problem (\cvca). Even for this simpler case, the best-known result is the aforementioned $2$-approximation due to Auletta et al.~\cite{Auletta99}, while the case of augmenting the edge-connectivity of a cycle by one admits much better approximation guarantees~\cite{Cecchetto20}, even via simple iterative algorithms~\cite{Galvez19}. Similarly to the case of edge-connectivity, it is possible to prove that \cvca is APX-hard (see Section~\ref{hardness_of_approximation}), so our focus will be on the design of approximation algorithms.

The following theorem summarizes the central result of this work.

\begin{theorem}
\label{thm:local_search}
	There is a $1.8704$-approximation for the cycle vertex-connectivity augmentation problem. 
\end{theorem}

Our algorithm consists basically of two phases: a first phase where we exhaust \emph{locally efficient} choices of links to be added to the instance, and then a second phase to complete the solution constructed so far. Roughly speaking, the cost of the partial solution constructed in the first phase can be properly bounded but does not ensure feasibility, issue that the second phase addresses. The number of links required in this second phase may be large, but in that case, the fact that no locally efficient set of links can be added to the solution implies that any feasible solution must include a considerable amount of links.

Our algorithm is similar in spirit to the one presented by Gálvez et al.~\cite{Galvez19} for the case of augmenting the edge-connectivity of a cycle; however, much more sophisticated tools are required for its analysis. In particular, the results from Gálvez et al.~heavily rely on the notion of \emph{contracting} links (i.e.~merging the endpoints of a link into a new super-node), since in the edge-connectivity case a solution is feasible if and only if iteratively contracting all its links turns the original cycle into a single super-node. This property, unfortunately, does not extend to the case of vertex-connectivity (see Figure~\ref{ejemplospasoslocales2} for an example of a set of links whose addition makes the graph $3$-edge-connected but not $3$-connected), and consequently, we develop alternative methods to verify the feasibility of a solution. By properly adjusting the parameters involved and using linear programming formulations to obtain refined lower bounds for the optimal solution, we are able to prove the claimed approximation guarantees. To the best of our knowledge, this is the first improvement for a non-trivial special case of the vertex-connectivity augmentation problem on $2$-vertex-connected graphs since Auletta et al.'s results.
\subsection{Related Results}

As mentioned before, the edge-connectivity augmentation problem has received considerable attention. A main reason is that it can be reduced to augmenting the edge-connectivity of either a tree or a cactus\footnote{A \emph{cactus} is a $2$-edge-connected graph where every edge belongs exactly to one cycle of the graph.} by one~\cite{Dinitz76}. The first case is known in the literature as the \textbf{Tree Augmentation} problem and many results have been developed for it. The problem is known to be APX-hard~\cite{Frederickson81,Cheriyan99,Kortsarz04}, and several better-than-2 approximation algorithms have been developed using a wide range of techniques~\cite{Nagamochi03,Even09,Cheriyan08,Kortsarz16,Adjiashvili19,Fiorini18,Cheriyan18,Kortsarz18,Grandoni18}, being $1.393$ the current best approximation ratio achieved by Cecchetto et al.~\cite{Cecchetto20}. Tree Augmentation has also been studied in the framework of Fixed-Parameter Tractability~\cite{Marx15,Basavaraju14} and in presence of general edge weights. In the latter case, for a long time the best approximation ratio was $2$~\cite{Frederickson81,Khuller93,Goemans94,Khuller1994}, but recently this approximation ratio has been improved in a major breakthrough by Traub and Zenklusen~\cite{TZ21}; prior to that, such an improvement had been achieved only for restricted cases~\cite{Cohen13,Adjiashvili19,Fiorini18,Grandoni18,Nutov17,Cecchetto20}.

The second case is known as the \textbf{Cactus Augmentation} problem, which is also APX-hard even if the cactus is a cycle~\cite{Galvez19}. For a long time the best-known approximation guarantee was $2$ due to Khuller and Vishkin ~\cite{Khuller1994}. Recently, Byrka et al.~\cite{Grandoni20} broke this barrier providing a $1.91$-approximation, which was further improved to $1.393$ by Cecchetto et al.~\cite{Cecchetto20}; this is consequently the current best approximation ratio for the general edge-connectivity augmentation problem. 

For the vertex-connectivity augmentation problem, it was first proved by Végh~\cite{Vegh11} that if all the possible chords are available as links, the problem can be solved in polynomial time for any $k$, but if the set of links is pre-specified then the problem becomes APX-hard~\cite{Kortsarz04}. It is worth noting that this result does not directly hold for the special case of cycles. For the case of vertex-connectivity augmentation of a $2$-connected graph, Auletta et al.~\cite{Auletta99} provide a $2$-approximation, while for the case of $k$-connected graphs the current best approximation ratio is $4 + \varepsilon$ for any fixed $k$, due to Nutov~\cite{Nutov20}. Recently, for the case of vertex-connectivity augmentation of $1$-connected graphs, Nutov developed a $1.91$-approximation~\cite{Nutov20-2node}.

\paragraph{\textbf{Organization of the paper.}} Section~\ref{sec:Prelim} provides some preliminary definitions and useful tools. In Section~\ref{sec:LocalSearch} we describe our main algorithmic approach and formally prove that the approximation ratio of our algorithm is strictly better than $2$. Then in Section~\ref{sec:improving} we describe how to enhance this algorithm and to refine the analysis so as to obtain the claimed approximation ratio and finally in Section~\ref{hardness_of_approximation} we prove that the problem is APX-hard. Some details and technical proofs are deferred to the Appendix.

\section{Preliminaries}\label{sec:Prelim}
In this work, we only consider simple graphs of $n\geq 4$ vertices. We let $C_n$ denote the cycle on the set $[n]=\{1, 2, \dots, n\}$ of vertices. Due to the cyclic aspect of $C_n$, sometimes we may use $n+1$ and $0$ to refer to vertices $1$ and $n$ respectively. A \textbf{chord} is an edge between non (cyclically) consecutive vertices of $C_n$. Each chord $ab$ divides the cycle into two nonempty \emph{sides} $V_1$ and $V_2$ so that $V(C_n) = V_1 \dot\cup V_2 \dot\cup \{a,b\}$. %$V_1$, $V_2$ and $\{a,b\}$ partitions $V(C_n)$. 
If $L$ is a set of chords, $V(L)$ will denote the set of vertices incident to at least one chord in $L$.

$C_n$ is $2$-connected, and consequently any graph obtained by adding chords to $C_n$ is $2$-connected as well. A \emph{separating pair} of a $2$-connected graph is a pair of vertices such that, if we remove them, the graph becomes disconnected. Thus, $3$-connected graphs are exactly those $2$-connected graphs without separating pairs. In what follows, a set $L$ of chords is said to $3$-connect the cycle if $C_n \cup L$ is 3-connected.

\begin{definition}
	Given a set of chords $S$ of $C_n$ such that $C_n \cup S$ is $3$-connected, the goal of the \textbf{cycle vertex-connectivity augmentation (\cvca) problem}  is to find a set $S'\subseteq S$ of minimum size that $3$-connects the cycle. We say that $(C_n, S)$ is an instance of \cvca.
\end{definition}

We reserve the word \emph{link} to refer to chords in $S$. Sets of links that $3$-connect the cycle are feasible solutions for \cvca.  
We start by introducing some required concepts to characterize these feasible solutions.

\begin{definition}
	Let $a$, $b$, $c$, $d$ be  distinct vertices in $C_n$ satisfying $a < b$ and $c < d$ (in the standard ordering of $[n]$). We say that a chord $ab$ \textbf{crosses} $cd$ if $c < a < d < b$ or $a < c < b < d$. 
\end{definition}

Equivalently, $ab$ and $cd$ cross if and only if $c$ and $d$ are on different sides of $ab$ (chords sharing a vertex do not cross). This definition is actually natural: If we draw $C_n$ in the plane as a circle, and $ab$ and $cd$ as straight open-ended line segments, then $ab$ crosses $cd$ if and only if the associated segments intersect.  

\begin{lemma} \label{Every_pair_must_be_crossed}
	$S'\subseteq S$ is feasible for \cvca if and only if every chord of $C_n$ is crossed by some link of $S'$. As a consequence, if $S'$ is a feasible solution, then every vertex of the cycle must be incident to some link and $|S'|\ge n/2$.
\end{lemma}

\begin{proof}
    The first claim follows since a chord of $C_n$ is a separating pair of $C_n\cup S'$ if and only if it is not crossed by any link of $S'$. For the second claim, notice that for each vertex $j$ there must be a link crossing the chord defined by vertices $(j-1)$ and $(j+1)$; since this link is incident to $j$, $S'$ must be an edge-cover and the bound follows. \end{proof}

Since any set of links containing a feasible solution is feasible, it makes sense to study minimal feasible solutions. A theorem by Mader \cite{Mader72} states that every cycle $Z$ of a $k$-connected graph $G$ that is formed by critical-edges (i.e., edges $e$ such that $G-e$ is not $k$-connected) contains a vertex whose degree in $G$ is $k$. This can be translated into the following useful property.

\begin{lemma}
\label{thm:minimal_solution}
	Every minimal solution $S'$ for \cvca is acyclic and thus has size at most $n-1$. In particular, every minimal set that $3$-connects the cycle defines a $2$-approximate solution for this problem.
\end{lemma}

\begin{proof}
    Let $C'$ be a cycle in a minimal solution $S'$. As all the vertices incident to $C'$ have degree equal to four in $C_n\cup S'$, Mader's theorem above implies that at least one link of $C'$ is not critical, contradicting the minimality of $S'$. 
    Consequently $S'$ is acyclic, and hence its size is at most $n-1$. Thanks to Lemma~\ref{Every_pair_must_be_crossed}, the optimum solution has size at least $n/2\geq |S'|/2$ and thus $S'$ is a $2$-approximate solution for our problem.
\end{proof}

Lemma~\ref{thm:minimal_solution} suggests a simple greedy $2$-approximation algorithm: Initialize $S'$ as $S$, and as long as there exists a link $e$ such that $S'-e$ is feasible, remove $e$ from $S'$. The resulting set $S'$ is a minimal solution and hence a $2$-approximate solution.

\subsection{Circle components}\label{sec:circle_comp}

The previous $2$-approximation was built in a top-down way, i.e., deleting links from a feasible solution. In order to obtain a better approximation ratio, we will instead build solutions in a bottom-up way. For that purpose, we use the concept of \emph{circle components} (see Figure~\ref{ejemplospasoslocales2}). 

\begin{definition}{\cite{Spinrad94}} The \textbf{circle graph} defined by a set of chords $L$ is the graph with vertex set $L$, where two chords are adjacent if and only if they cross. 
    
\end{definition}

\begin{definition}
    A \textbf{circle component} is a set of links $L \subseteq S$ such that the circle graph defined by $L$ is connected. 
    A \textbf{link path} from vertex $a$ to vertex $b$ is a sequence $e_1e_2\dots e_k$ of links such that $a$ is an endpoint of $e_1$, $b$ is an endpoint of $e_k$ and each link crosses the next one in the sequence. 
\end{definition}

\begin{figure}[t]
	\begin{center}
		\resizebox{.7\textwidth}{!}{\includegraphics[scale=.85]{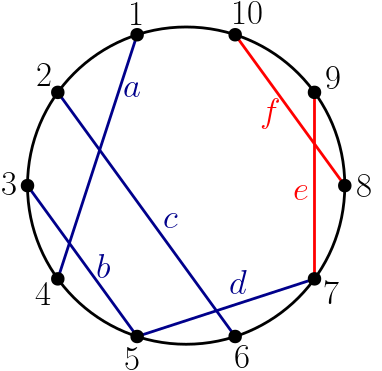}\hspace{80pt}\includegraphics[scale=.99]{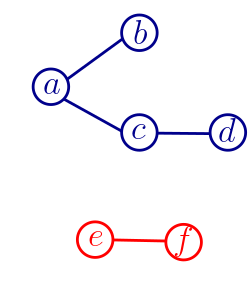}}		
	\end{center}
	\caption{\textbf{Left:} Two circle components in blue and red. With respect to the blue component, vertices $2,3,4,5$ and $6$ are internal, vertices $1$ and $7$ are border and the remaining vertices are external; with respect to the red component, vertices $8$ and $9$ are internal while vertices $7$ and $10$ are border. The sequence $bac$ is an example of a link path between $5$ and $2$.
	\textbf{Right:} The associated circle graph. Notice that it is not connected as, for instance, vertices $1$ and $7$ form a separating pair.}
	\label{ejemplospasoslocales2}
\end{figure}

Observe that if $L$ is a circle component, then for every $a, b \in V(L)$ there is a link path in $L$ from $a$ to $b$ (see Figure~\ref{ejemplospasoslocales2}). The next technical definition provides further useful concepts we will use along this work.

\begin{definition}
    Let $L$ be a circle component and $V(L)=\{v_1, v_2,\cdots, v_{|V(L)|}\}$ such that $v_1 < v_2 < \cdots < v_{|V(L)|}$ in the cycle. We say that $v_i,v_{i+1}$ (indices considered modulo $|V(L)|$) are consecutive in $V(L)$.
    
   Chords between non-consecutive vertices of $V(L)$ are called \textbf{internal chords of $L$}. Chords between consecutive vertices in $V(L)$ are called  \textbf{border chords}. Vertices of $V(L)$ incident to some border chord are called \textbf{border vertices} and the rest of the vertices of $V(L)$ are called \textbf{internal vertices}. The vertices outside $V(L)$ are called \textbf{external vertices}.
\end{definition}

Notice that an internal chord of $L$ may connect two internal vertices, one internal vertex and a border vertex or even two border vertices (if they are incident to different border chords). We also say that $L$ crosses a chord $ab$ of $C_n$ if some link  $e\in L$ crosses $ab$. 

The following lemma contains the main characterization of feasibility we will use along this work.

\begin{lemma} \label{circle-component-characterization}
Let $S'\subseteq S$ be an edge-cover of $C_n$. The set $S'$ is a feasible solution for \cvca if and only if $S'$ is a circle component.
\end{lemma}

We need first the following technical lemmas.

\begin{lemma} \label{extremes_in_chord_path_are_crossed_by_P}
Let $ab$ and $cd$ be chords of $C_n$ and $W=e_1\dots e_k$ be a link path from $a$ to $b$. If $ab$ crosses  $cd$ then there is some chord in $W$ that crosses $cd$.
\end{lemma}

\begin{proof}
Let $V_1$ and $V_2$ be the sets of vertices of $C_n$ on each side of the chord $cd$, so that $a\in V_1$ and $b\in V_2$. Observe also that $V_1\cup \{c,d\}$ is an interval in the cycle with extremes $c$ and $d$.
Let $e_i$ be the first chord in $W$ with at least one endpoint in $V_2$ (this link exists since $e_k$ has one endpoint in $V_2$). If $i=1$ then $e_1$ crosses $cd$, so assume $i\geq 2$. Since the endpoints of $e_{i-1}$ are in $V_1\cup \{c,d\}$ which is an interval, one of the sides of $e_{i-1}$ is fully contained in $V_1$. Therefore, as $e_{i-1}$ crosses $e_i$, one endpoint of $e_i$ is also in $V_1$. We conclude that $e_i$ has one endpoint in each side of $cd$ and so $e_i$ crosses $cd$.
\end{proof}

\begin{lemma} \label{general_crossing} Let $L$ be a circle component and $ab$ a chord. $L$ crosses $ab$ if and only if (1) $ab$ is an internal chord of $L$ or (2) $ab$ crosses some border chord $cd$ of $L$.
\end{lemma}

\begin{proof}
We proceed by cases depending on $ab$.

\paragraph{Case 1: $ab$ is an internal chord of $L$.} Let $V_1$ and $V_2$ be the two sides defined by $ab$. Since $ab$ is internal, $V(L)$ has vertices on both sides, so let $c\in V(L)\cap V_1$, $d\in V(L)\cap V_2$ and $W$ be a link path in $L$ from $c$ to $d$. Since $cd$ crosses $ab$, we conclude by using Lemma \ref{extremes_in_chord_path_are_crossed_by_P} that there exists a link of $W\subseteq L$ that crosses $ab$.
\paragraph{Case 2: $ab$ crosses some border chord $cd$ of $L$.} Let $W$ be a link path in $L$ from $c$ to $d$. By Lemma \ref{extremes_in_chord_path_are_crossed_by_P}, some link in this path must cross $ab$.
\paragraph{Case 3: $ab$ is a border chord of $L$}
Since $a$ and $b$ are consecutive in $L$ there is no link in $L$ with an extreme in one of the sides of $ab$ and so $L$ does not cross $ab$.
\paragraph{In any other case.} At least one extreme (say $a$) of the chord is external to $L$. Starting from $a$, let $c$ (respectively $d$) be the first vertex of $V(L)$ that we encounter going clockwise (respectively counterclockwise) along the cycle. Then $cd$ is a border chord of $L$. Let $V_1$ be the side of $cd$ to which $a$ belongs and note that $V_1\cap V(L)=\emptyset$. Since $ab$ does not cross $cd$ (see case 2), we conclude that $b$ is in the interval $V_1\cup \{c,d\}$. Therefore, one side of $ab$ is completely contained in $V_1$, and thus this side does not contain any vertex of $V(L)$. It follows that no link in $L$ crosses $ab$.
\end{proof}

We now have all the ingredients to prove Lemma~\ref{circle-component-characterization}.

\begin{proof}[Proof of Lemma \ref{circle-component-characterization}]
Suppose that $S'$ is a circle component. Due to Lemma~\ref{Every_pair_must_be_crossed}, every chord $ab$ of $C_n$ is internal for $S'$, and thus, by Lemma \ref{general_crossing}, $ab$ is crossed by some link in $S'$. Lemma~\ref{Every_pair_must_be_crossed} implies that $S'$ is feasible for \cvca.

For the converse, let $S'$ be a feasible solution of $\cvca$. Suppose by contradiction that the circle graph $G'$ of $S'$ is not connected and let $L$ be a connected component of $G'$. Note that $L$ cannot be a singleton, since if $L=\{ab\}$ then by hypothesis there exists some chord of $S'$ crossing $ab$, and so $L$ would not be maximally connected in $G'$. It follows that $L$ is a circle component. 

Suppose first that $L$ has a border chord $ab$. By Lemma \ref{general_crossing}, $L$ does not cross $ab$.
Since $S'$ crosses $ab$, there exists some link $cd\in S'\setminus L$ such that $cd$ crosses $ab$. Lemma \ref{general_crossing} then implies that $L\cup \{cd\}$ is also a circle component, contradicting the maximality of $L$. We conclude that all chords with endpoints in $V(L)$ are internal chords of $L$. The only way for this to happen is that $V(L)=[n]$. But then, every chord of the cycle is internal to $L$. In particular, any $cd\in S'$ is crossed by $L$, and by maximality of $L$, $cd\in L$. We conclude that $S'=L$ and so $S'$ is a circle component.
\end{proof}
    
\subsection{Minimal completions}\label{sec:minimal_comp}

Another tool we will make use of along this work is a procedure to turn a set of circle components into a feasible solution while controlling the number of extra links.

\begin{definition}
    Let $F\subseteq S$ be a set of links. We call a set $Q \subseteq S$ a \textbf{completion} of $F$ if $F \cup Q$ is feasible for \cvca. 
    A \textbf{minimal completion of $F$} is a completion that is minimal for inclusion.
\end{definition}

We say that two circle components $L_1$ and $L_2$ cross if there exists $e_1 \in L_1$ and $e_2\in L_2$ such that $e_1$ crosses $e_2$. The main idea of our algorithm in Section~\ref{sec:LocalSearch} is to build a collection $\mathcal{L}$ of circle components that, roughly speaking, use few links to cover many vertices and are pairwise non-crossing; then we add a minimal completion to obtain the desired feasible solution. A similar approach has been used by Gálvez et al. for the case of edge-connectivity~\cite{Galvez19}, where the authors prove that any set of links can be completed by iteratively picking and contracting links so as to reduce the whole graph into a single super-node. As mentioned before, this procedure unfortunately does not work for the case of vertex-connectivity as now we require instead to merge the different circle components into a single circle component, which is in general a strictly stronger requirement. The following lemma provides a way to compute a minimal completion in the case of vertex-connectivity and also allows us to bound its size. 

\begin{lemma} \label{minimal_completions_of_components}
Let $L_1,\dots, L_k$ be $k\geq 1$ non-crossing circle components, $|L_i|\ge 2$ for each $i=1,\dots,k$, and let $F=L_1\cup \dots \cup L_k$. Then every minimal completion $Q$ of $F$ uses at most $n - 3 - \sum_{i=1}^k (|V(L_i)| - 3)$ links. Furthermore, such a minimal completion can be computed in polynomial time.
\end{lemma}

The proof consists of two steps: We first prove the claim for the simpler case of a single circle component, and then reduce the general case to this simpler one by introducing the notion of \emph{zone graph} (see Figure~\ref{ejemplospasoslocales3}).

\paragraph{Step 1. Completing a single circle-component.}

\begin{lemma}\label{minimal_completion}
Let $L \subseteq S$ be a circle component. Then, every minimal completion $Q$ of $L$ has size at most $n - |V(L)|$.
\end{lemma}

\begin{proof}
Let us first prove that there exists a completion $Q$ of size $n-|V(L)|$. We do this by induction on $i=n-|V(L)|$. If $i=0$ then $L$ is an edge-cover and by Lemma \ref{circle-component-characterization}, it is already feasible, so we can set $Q=\emptyset$. For $i \geq 1$ we notice that there must be a border chord $ab$ of $L$. Since $S$ itself is feasible, there must be a link $cd\in S$ that crosses $ab$. Then, by Lemma \ref{general_crossing}, $L'=L\cup\{cd\}$ is a component covering at least one more vertex than $L$. By induction, there is a completion $Q'$ of $L'$ with at most $n-|V(L')|$ links. Then $Q=Q'\cup \{cd\}$ is a completion of $L$ satisfying that $|Q|=|Q'|+1\leq 1 + n - |V(L')| \leq n - |V(L)|$, concluding the first proof.

Now, let $Q$ be a minimal completion of $L$. Consider the new instance of \cvca $(C_n, L\cup Q)$. In this instance, $Q$ is the only completion of $L$, and then following the previous argumentation we have that $|Q|\leq n-|V(L)|$. 
\end{proof}

It is worth noting that Lemma~\ref{minimal_completion} provides an alternative proof that minimal solutions are $2$-approximate not relying on Mader's theorem (see Lemma~\ref{thm:minimal_solution}) as the following proposition states.

\begin{proposition}\label{corollary_completion}
Every minimal completion of a single chord has at most $n-3$ links, and every minimal solution for $\cvca$ has size at most $n-2$.
\end{proposition}

\begin{proof}
    Let $Q$ be a minimal completion of a single link $\{e\}$. Since $Q\cup \{e\}$ is a solution for $\cvca$, there exists $f\in F$ that crosses $e$. Then $Q\setminus \{e,f\}$ is a minimal completion of the circle component $\{e,f\}$. Using Lemma \ref{minimal_completion}, we get $|Q|=1+|Q\setminus\{f\}| \leq 1+n-|V(\{e,f\})|=1+n-4=n-3$.

    For the second statement, we observe that if $Q$ is a minimal solution for $\cvca$ then $Q-e$ is a minimal completion of $\{e\}$ for any link $e\in Q$, and so by the previous paragraph $|Q|=1+|Q-e|\leq n-2$. \end{proof}

\paragraph{Step 2. Completing a collection of arbitrary non-crossing circle components.} 

As mentioned before, we will relate the case of arbitrary non-crossing circle components to the case of a single circle component via the following definitions.

\begin{definition}
The \textbf{open zone} defined by a border chord $ab$ of a circle component $L$ is the interval of vertices $I(ab, L)$ corresponding to the side of $ab$ that does not include any vertex of $V(L)$. The associated closed zone is $I(ab,L)\cup \{a,b\}$.

    The \textbf{zone graph} defined by a border chord $ab$ of a circle component $L$ is the graph $Z(ab,L,S)$ obtained from $G=([n],C_n \cup S)$ by doing the following operations: (1) Erase all the links of $L$, (2) Add a link $e$ between $a$ and $b$, and (3) Contract all the vertices in $C_n\setminus (I(ab,L)\cup \{a,b\})$ to a single vertex $v_0$.
\end{definition}

\begin{figure}[t]
	\begin{center}
	    \includegraphics[scale=.6]{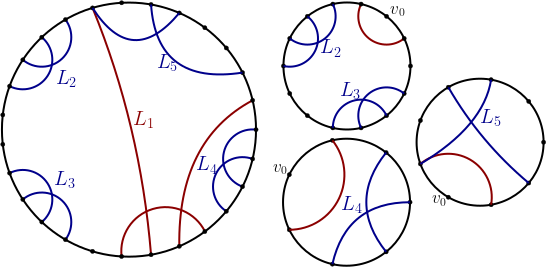}
	\end{center}
	\caption{\textbf{Left: }Graph with $5$ components. $L_1$ defines 3 zones, which can have one component ($L_4$) Two or more components (On the zone that includes $L_2$ and $L_3)$ and extra vertices (like the ones next to $L_3$). \textbf{Right:} The corresponding Zone Graphs}
	\label{ejemplospasoslocales3}
\end{figure}

In $Z=Z(ab,L,S)$ we have removed all the loops and kept only one link from each parallel class. For every set $Q\subseteq S$ of chords in $G$, we let $\psi(Q)$ be the associated set of chords in $Z$, eliminating every link that becomes part of the external cycle. By this construction, we obtain a new \emph{zone instance} $(C_{n'}, \psi(S) \cup \{ab\})$, associated to $ab$ and $L$. Here we note that $n'=|I(ab,L)|+1<n$ which follows since $L$ covers at least 4 vertices. 

Let $\mathcal{L}$ be a collection of non-crossing circle components, and let $L\in \mathcal{L}$ be a fixed one. Let $P(L)$ be the border chords ($P$ stands for perimeter) of $L$. Since they do not cross, every circle component $K\in \mathcal{L}\setminus \{L\}$ is completely contained in one closed zone of $L$ (meaning that $V(K)$ is a subset of the closed zone), and every closed zone of $L$ may contain further components of $\mathcal{L}\setminus \{L\}$. For each border chord $ab\in P(L)$, let $\mathcal{L}(ab,L)$ be the set of circle components in $\mathcal{L}\setminus \{L\}$ that are in the closed zone $I(ab,L)\cup \{a,b\}$, and let $F^{ab} = \bigcup_{K\in \mathcal{L}(ab,L)} K$ be the sets of links that are contained in some circle component of the closed zone $I(ab,L)\cup \{a,b\}$. We can then prove the following.

\begin{lemma} \label{zone_instance} Let $Q\subseteq S$ be a set of links.
If $Q$ is a completion of $F^{ab}\cup L$ in the original instance $(C_n,S)$, then $\psi(Q)$ is a completion of $\psi(F^{ab}) \cup \{ab\}$ in the zone instance $(C_{n'}, \psi(S)\cup \{ab\})$. Conversely, if $\psi(Q)$ is a completion of $\psi(F^{ab}) \cup \{ab\}$ in the zone instance, then $Q\cup F^{ab}$ crosses all the chords of $C_n$ with both extremes in $I(ab,L) \cup \{a,b\}$.
\end{lemma}

\begin{proof}
Let $Q$ be a completion of $L$ in the original instance and let $cd$ be a chord in $C_{n'}$. If both $c$ and $d$ are uncontracted vertices (i.e., they are in the closed zone $I(ab,L)\cup \{a,b\}$), then they are not crossed by $L$ in the original instance, and so they must be crossed by some link $e \in Q\cup F^{ab}$. Therefore, the respective $\psi(e) \in \psi(Q)\cup \psi(F^{ab})$ crosses $cd$ in the zone instance. On the other hand, if one of the vertices of the chord $cd$ is $v_0$, then $cd$ is crossed by $ab$. In any case, $cd$ is crossed by $\psi(Q)\cup \psi(F^{ab})\cup \{ab\}$.

For the converse, suppose that $\psi(Q)$ is a completion of $\psi(F^{ab}) \cup \{ab\}$ in the zone instance. We notice that all the vertices in $I(ab,L)\cup \{a,b\}$ are uncontracted vertices. Then for every chord $cd$ in $C_n$ between two vertices in that set, there is a link $f'$ in $\psi(Q)\cup \psi(F^{ab})$ that crosses the chord in $C_{n'}$. Any link $f \in Q\cup F^{ab}$ with $\psi(f)=f'$ crosses the chord $cd$ in the original instance.
\end{proof}

We can now proceed with the proof of Lemma~\ref{minimal_completions_of_components}

\begin{proof}[Proof of Lemma~\ref{minimal_completions_of_components}]
 We will prove the claim by induction on $n$. Let $Q$ be a minimal completion of $F=L_1\cup \dots \cup L_k$. The case $n=4$ holds trivially since in that case $k=1$ and $L_1$ covers all the vertices and so $\emptyset$ is the sought minimal completion. 
    
 Consider a border chord $ab\in P(L_1)$. Lemma \ref{zone_instance} guarantees that $\psi(Q)$ is a completion of $\psi(F^{ab})\cup \{ab\}$  in the associated zone instance.

    Therefore, $\psi(Q)\cup \{ab\}$ is a completion of $\psi(F^{ab})$ in the same zone instance. So,     there is a subset $Q'_{ab}\subseteq Q$ with $|Q'_{ab}|=|\psi(Q'_{ab})|$ such that either $\psi(Q'_{ab})$ or  $\psi(Q'_{ab})\cup \{ab\}$ are minimal completions of $\psi(F^{ab})$ in that zone instance (in any case, $Q'_{ab}$ has size at most that of a minimal completion of $\psi(F_{ab})$ in the zone graph).

    We have two cases here: let us consider first the case in which $\mathcal{L}(ab,L_1)$ is non-empty. Recall that the zone instance has $n_{ab}<n$ vertices, so we can apply the induction hypothesis to the zone instance and obtain that 
    $$|Q'_{ab}| \leq n_{ab}-3 - \sum_{J \in \mathcal{L}(ab,L_1)}(|V(J)| - 3) = |I(ab,L_1)| - \sum_{J \in \mathcal{L}(ab,L_1)}(|V(J)| - 3|).$$
    
    In the complementary case when $\mathcal{L}(ab,L_1)$ and $F^{ab}$ are empty, i.e. when there are no components in the zone defined by $ab$, we have that $\psi(Q)$ is a completion of a single link $\{ab\}$ in the zone instance, and so there is a set $Q'_{ab}\subseteq Q$ with $|Q'_{ab}|=|\psi(Q'_{ab})|$ such that $\psi(Q'_{ab})$ is a minimal completion of $\{ab\}$. By Proposition~\ref{corollary_completion}, $|Q'_{ab}|=|\psi(Q'_{ab})|=n_{ab}-3=|I(ab,L_1)|=|I(ab,L_1)|-\sum_{J \in \mathcal{L}(ab,L_1)}(|V(J)|-3)$ since the sum is empty.
    
    Define $Q' = \bigcup_{ab\in P(L_1)} Q'_{ab}$. We claim that $Q'$ is a completion of $F=L_1\cup \dots \cup L_k$ in the original instance. Indeed, let $cd$ be an arbitrary chord of the cycle $C_n$ and let us prove that $Q'\cup F$ crosses it. If $cd$ is crossed by $L_1$ then we are done, so assume that it is not crossed by it. By Lemma \ref{general_crossing}, $cd$ is neither an internal chord of $L_1$ nor it is a chord that crosses some border chord of $L_1$. The only possibility left is that $cd$ connects two vertices in  $I(ab, L_1)\cup \{a,b\}$ for some border chord $ab\in P(L_1)$.
    Lemma \ref{zone_instance} implies that $cd$ is crossed by a link in $Q'_{ab}\cup F^{ab}\subseteq Q'\cup F$ and thus the claim is fulfilled.
    
    Finally, as $Q' \subseteq Q$ and $Q$ is a minimal completion of $F$, we must have that $Q'=Q$. Using the bounds above and the fact that the open zones $\{I(ab,L_1)\colon ab\in P(L_1)\}$ form a partition of $[n]\setminus V(L_1)$, we get that 
    \begin{align*}
    |Q|&\leq \sum_{ab\in P(L_1)} |Q'_{ab}| \leq \sum_{ab\in P(L_1)} \bigl(|I(ab,L_1)|-\sum_{J\in \mathcal{L}(ab,L_1)}(|V(J)|-3)\bigr)\\
    &=(n-|V(L_1)|) - \sum_{J \in \mathcal{L}\setminus \{L_1\}}(|V(J)|-3)  =n - 3 - \sum_{J \in \mathcal{L}}(|V(J)|-3).
    \end{align*}
\end{proof}

In what follows we will provide a method to compute a partial solution with good enough guarantees so as to apply Lemma~\ref{minimal_completions_of_components} and complete it. Efficient choices of links will be prioritized in this first phase in order to obtain an approximation ratio better than $2$.

\section{Local Search Algorithm}\label{sec:LocalSearch}

In this section we present our main algorithmic approach for \cvca. Let us first define special sets of links which will be fundamental for our algorithm.

\begin{definition}
    Let $F \subseteq S$ be a set of links. We say that $F$ is \textbf{singleton-free} if the circle graph $G'$ of $F$ has no isolated vertices (singletons). The connected components of $G'$ form a family $\mathcal{L}$ of non-crossing circle components. 
    We define the \textbf{utility of $F$} as \[U(F)=-|F|+\sum_{J\in \mathcal{L}}(|V(J)|-3).\]
\end{definition}

Our approach is to initially find singleton-free sets in an increasing way (i.e., by only adding links) whose utility is high with respect to some increasing parameter (e.g. the number of covered vertices). As the following lemma shows, high utility sets indeed provide solutions with improved approximation guarantees.

\begin{lemma} \label{bound_alg}
If $Q$ is a minimal completion of a singleton-free set $F$, then $Q\cup F$ is a solution for $\cvca$ with at most $n-3-U(F)$ links.
\end{lemma}
\begin{proof}
By Lemma \ref{minimal_completions_of_components},  $Q$ has size at most $n-3-\sum_{J\in \mathcal{L}}(|V(J)|-3)=n-3-|F|-U(F)$, and so $|Q|+|F|\leq n-3-U(F)$.
\end{proof}

For any given $\alpha \in (1/2,1]$ and any fixed step size $N_{\max}\in \mathbb{N}$, we define the following Local Search algorithm which works in two phases: in the first phase, it constructs a singleton-free set of links $F$ by adding at most $N_{\max}$ links at a time and ensuring that, in each iteration, the marginal utility gain is at least $(1-\alpha)$ times the number of newly covered vertices. More in detail, the algorithm adds links in the first phase as long as the set $F$ under construction is not $(\alpha, N_{\max})$-critical as defined below.

\begin{definition}
Let $\alpha \in (1/2,1]$ and let $N_{\max}\in \mathbb{N}$. A singleton-free set $F\subseteq S$ will be called $(\alpha, N_{\max})$-critical if there is no set $K\subseteq S\setminus F$ of size at most $N_{\max}$ such that $F\cup K$ is singleton-free and $U(F\cup K)-U(F)\geq (1-\alpha)|V(F\cup K)\setminus V(F)|$. \end{definition}

Notice that at the end of the first phase, the utility of $F$ is at least $(1-\alpha) |V(F)|$. In the second phase, the algorithm finds a minimal completion $Q$ of $F$ and returns the pair $(Q,F)$. The set $Q\cup F$ is the solution proposed. See Algorithm~\ref{LocalSearchAlgorithm}.

\begin{algorithm}
\caption{Local Search algorithm}\label{LocalSearchAlgorithm}
\hspace*{\algorithmicindent} \textbf{Input:} Instance $(C_n,S)$ of $\cvca$. $\alpha \in (1/2, 1)$, $N_{\max} \in \mathbb{N}$.
\begin{algorithmic}[1]
\phase{Building a collection of non-crossing circle components.}
\State $F\gets \emptyset$.
\While {$F$ is not $(\alpha,N_{\max})$-critical}
\State Let $K\subseteq S\setminus F$ of size at most $N_{\max}$ such that $F\cup K$ is singleton-free and $U(F\cup K)-U(F)\geq (1-\alpha)|V(F\cup K)\setminus V(F)|$;
\State $F\gets F\cup K$
\EndWhile
\phase{Finding a completion}
\State Find a minimal completion $Q$ of $F$.
\State Return $(Q,F)$
\end{algorithmic}
\end{algorithm}

To gain some intuition about the algorithm, suppose that at one iteration of Phase 1, the algorithm finds a set $K$ of at most $N_{\max}$ links to add to the current solution $F$. Let $\mathcal{L}$ be the current circle components of $F$. The simplest case is when the set $K$ crosses exactly one circle component $L\in \mathcal{L}$ and $V(K)$ does not intersect any other component of $\mathcal{L}$. In this case, the marginal utility gain $U(F\cup K)-U(F)$ is simply $-|K|+|V(L\cup K)\setminus V(L)|=-|K|+|V(F\cup K)\setminus V(F)|$.
So, the algorithm adds $K$ to $F$ as long as $-|K|+|V(F\cup K)\setminus V(F)|\geq (1-\alpha)|V(F\cup K)\setminus V(F)|$, or equivalently if the ratio from the number of new links $K$ to the number of new covered vertices $|V(F\cup K)\setminus V(F)|$ is at most $\alpha$. So, the algorithm only adds sets $K$ of links that cover new vertices with few links.

The intuition behind adding sets $K$ that cross (or touch) more than one component is more difficult to grasp, but as a general rule, connecting more components yields larger marginal utility gains. The following lemma provides a bound on the size of the returned solution with respect to the number of nodes covered by links from Phase $1$.

\begin{lemma} \label{bound_for_alg}
For any fixed $N_{\max}$ and $\alpha$, the Local Search algorithm runs in polynomial time and returns a pair $(Q,F)$ such that $Q\cup F$ is feasible for $\cvca$ and $|Q\cup F|\leq n-3-(1-\alpha)|V(F)|$.
\end{lemma}
\begin{proof} Let $s=|S|$. Every iteration of the first phase may be achieved by trying all $s^{O(N_{\max})}$ possible subsets of $S\setminus L$ of size at most $N_{\max}$, hence taking polynomial time. Finding a minimal completion $Q$ can also be done in polynomial time by iteratively removing links from $S\setminus F$ that are not needed.

Let $q$ be the number of iterations of the first phase and let $F_i$ be the set $F$ at the end of the $i$-th iteration (using $F_0=\emptyset$). Then, by a telescopic sum and using the fact that $V(F_{i-1})\subseteq V(F_i)$ for each $i$, we have 
\[\textstyle U(F)=\sum_{i=1}^q U(F_i)-U(F_{i-1}) \geq (1-\alpha)\sum_{i=1}^{q}|V(F_i)\setminus V(F_{i-1})| = (1-\alpha)|V(F)|.\] By Lemma \ref{bound_alg}, $Q\cup F$ is feasible and has size at most $n-3-(1-\alpha)|V(F)|$.
\end{proof}

Roughly speaking, using that the optimal solution has size at least $n/2$ (by Lemma~\ref{Every_pair_must_be_crossed}) the previous lemma guarantees good approximation factors as long as $|V(F)|$ is large enough (where $F$ is the $(\alpha, N_{\max})$-critical set at the end of Phase 1). On the other hand, we will prove that whenever $|V(F)|$ is small there are stronger lower bounds on the size of the optimal solution (see Figure~\ref{fig:intuition}). 

\begin{figure}[t]
	\begin{center}
		\resizebox{.68\textwidth}{!}{\includegraphics[scale=0.4]{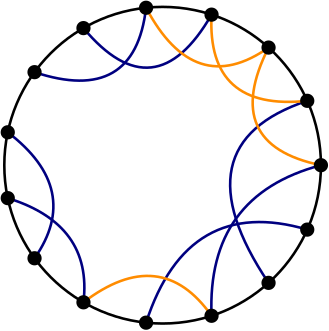}\hspace{70pt}\includegraphics[scale=0.52]{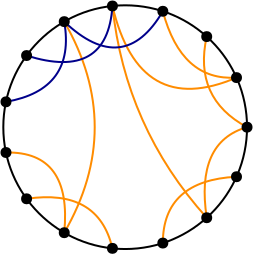}}		
	\end{center}
	\caption{Depiction of cases arising when applying the Local Search algorithm. Blue links correspond to $F$ while orange links correspond to a minimal completion. \textbf{Left:} When $|V(F)|$ is large, many vertices are covered by sets of links of high utility which pay for the size of the minimal completion.
	\textbf{Right:} When $|V(F)|$ is small, since $F$ is $(\alpha, N_{\max})$-critical, any set of links in $S\setminus F$ must have low utility, implying stronger lower bounds for $|OPT|$.}
	\label{fig:intuition}
\end{figure}

\subsection{Properties of Critical Sets}

Let $F$ be a fixed $(\alpha, N_{\max})$-critical set, with $N_{\max}\geq 1$ (for instance, the set $F$ at the end of Phase 1 of Algorithm~\ref{LocalSearchAlgorithm}). We also define $V'=V(C_n)\setminus V(F)$ and the set $\mathcal{L}$ of the non-crossing circle components associated to $F$. A first property we prove is that any link in $S\setminus F$ crosses at most one component $L\in \mathcal{L}$ and, if it does, it crosses at most one border chord of $L$. In particular, at least one endpoint will be in $V(L)\subseteq V(F)$. This intuitively implies that any relatively small set of remaining links in the instance cannot cover many vertices of the cycle and cross $F$ simultaneously (as in particular at least one of the endpoints of each such link will be in $V(L)\subseteq V(F)$). In what follows $P(J)$ denotes the set of border chords of a circle component $J$, and $P(F) := \bigcup_{J\in \mathcal{L}} P(J)$ denotes the \emph{perimeter} of $F$.

\begin{lemma} \label{links_can_cross_at_most_one_chord} Any link $e\in S$ can cross at most one chord in $P(F)$. If both endpoints of $e$ are in $V'$, then $e$ cannot cross any chord in $P(F)$. 
\end{lemma}

\begin{proof} Suppose that $e$ crosses at least one border chord in $P(F)$. Then, by Lemma \ref{general_crossing}, it crosses at least one circle component of $\mathcal{L}$ and it can furthermore be proved that it has to cross exactly one (see Lemma~\ref{atmost1} in the Appendix). So let $L\in \mathcal{L}$ be the component crossed by $e$. 

Note that $e$ can cross at most 2 border chords of $L$, and the only way for this to happen is that the extremes of $e$ are outside $V(L)$ and in different zones of $L$. In this case, $U(F\cup \{e\})-U(F)=-1+|V(L\cup \{e\})|-|V(L)|=1\geq |V(F\cup \{e\})-V(F)|(1-\alpha)$ which is not possible since $F$ was $(\alpha,N_{\max})$-critical. 

From here, any link $e$ in $S$ can cross at most one chord in $P(F)$. Furthermore, if both endpoints of $e$ are in $V'$ and $e$ crosses a chord in $P(F)$, then it would have to cross two by the previous paragraph, which is not possible. Hence, in the latter case, $e$ cannot cross any chord in $P$. \end{proof}

Another way to obtain lower bounds on the size of the optimal solution is to show that any maximal matching $M\subseteq \OPT[V']$ outside $V(F)$ must be large (since they cannot induce large circle components). We will argue about the number of links required to \emph{connect} $M$ as specified in the following definition.

\begin{definition}
    We say that a link $e \in S$ \textbf{connects} a set of links $X$ if $X \cup\{e\}$ forms a circle component.
\end{definition}

Since $F$ is $(\alpha, N_{\max})$-critical, $M$ must be formed by a considerable amount of circle components as mentioned before, which need to be connected so as to get a feasible solution; however, again since $F$ is $(\alpha, N_{\max})$-critical, each one of the links from the optimal solution required to connect these circle components cannot cross too many links from $M$ simultaneously. As a consequence, the number of such links cannot be too small, providing then a lower bound on the size of the optimal solution. The following lemma formalizes this previously described intuition and provides the required bounds to analyze the performance of the algorithm. 

\begin{lemma} \label{e_crosses_at_most_k_edges_in_M}
Suppose that $N_{\max}\geq \lceil (5-2\alpha)/(2\alpha -1)\rceil+1$. 
Let $M\subseteq S$ be any matching consisting of links with both endpoints in $V':=V(C_n)\setminus V(F)$. For any link $e\in S\setminus (M\cup F)$, the number of links from $M$ that $e$ connects is at most $\max\left\{0, \displaystyle \left\lceil \frac{5-2X(e)-V_{F}(e)- (4-V_{M}(e)-V_{F}(e))\alpha}{2\alpha -1} \right\rceil\right\}$, where
$X(e)$ is the indicator that $e$ crosses $F$, $V_M(e)=|V(\{e\})\cap V(M)|$ and $V_{F}(e)=|V(\{e\})\cap V(F)|$. In particular, no link of $S\setminus (M\cup F)$ connects more than $\ell:=\lceil (5-2\alpha)/(2\alpha-1)\rceil$ links in $M$.
\end{lemma}

\begin{proof} Let $\gamma = \max\left\{0, \displaystyle \left\lceil \frac{5-2X(e)-V_{F}(e)- (4-V_{M}(e)-V_{F}(e))\alpha}{2\alpha -1} \right\rceil\right\}$ be the bound in the Lemma. Note that we always have $\gamma\leq \ell \leq N_{\max}-1$.
   
Let $M_e$ be the set of links of $M$ that $e$ connects (this set contains the set of links of $M$ that $e$ crosses). Suppose for the sake of contradiction that $|M_e|> \gamma$. If $|M_e|\geq N_{\max{}}-1$ then redefine $M_e$ as any subset of $M_e$ with $N_{\max{}}-1$ links that $e$ connects (this can be done by pruning leaves in a covering tree of the circle graph of $M_e\cup \{e\})$. Let also $V_{M_e}(e)=|V(\{e\}\cap V(M_e)|$.
    
    Notice that $2\leq \gamma + 2 \leq |M_e\cup \{e\}|\leq N_{\max{}}$. We also have that $F\cup M_e \cup \{e\}$ is singleton-free. Furthermore $|V(M_e\cup \{e\}\cup F)\setminus V(F)|= 2|M_e|+|V(\{e\})\setminus V(M_e\cup F)|=2|M_e|+2-V_{M_e}(e)-V_{F}(e).$ Then, since $F$ is $(\alpha, N_{\max})$-critical, 
    we have that 
    $$U(M_e \cup \{e\} \cup F) - U(F) < (1-\alpha)( 2|M_e|+2-V_{M_e}(e)-V_{F}(e)).$$
    
    If $e$ does not cross $F$ then $M_e\cup \{e\}$ is a circle component not crossing $F$. Then \begin{eqnarray*}U(M_e\cup \{e\}\cup F)-U(F) & = & -|M_e\cup \{e\}|+|V(M_e\cup \{e\})|-3 \\ & = & -|M_e|-1 + (2|M_e|+2-V_{M_e}(e)) -3 \\ & = &  |M_e|-V_{M_e}(e)-2. \end{eqnarray*}

    If $e$ crosses $F$ then it does on a single circle component $L$, and we know that exactly one endpoint of $e$ is in $V(L)$. 
    Let $v$ be the other endpoint of $e$, and note that $V_{M_e}(e)=1$ if $v\in V(M_e)$ and $V_{M_e}(e)=0$  if $v\not\in V(M_e)$.
    Therefore, \begin{eqnarray*} U(M_e\cup \{e\}\cup F)-U(F) & = & -|M_e\cup \{e\}| + |V(M_e\cup \{e\}\cup L)|-|V(L)| \\ & = & -|M_e|-1+|\{v\}\cup V(M_e)| \\ & = & -|M_e|-1+2|M_e|+(1-V_{M_e}(e)) \\ & = & |M_e|-V_{M_e}(e). \end{eqnarray*}
    
    Summarizing we always have 
    $$|M_e|-V_{M_e}(e)+2-2X(e)<(1-\alpha)(2|M_e|+2-V_{M_e}(e)-V_{F}(e))$$
    
    Solving the inequality, we get
    $
        |M_e|<\frac{4-2X(e)-V_{F}(e)-\alpha(2-V_{M_e}(e)-V_{F}(e))}{2\alpha -1}
    $. Since $V_{M_e}(e)\leq V_M(e)$, the inequality above holds if we replace $V_{M_e}(e)$ by $V_M(e)$.
    
    Furthermore, since $|M_e|$ is an integer, this means that 
    \begin{align*}
        |M_e|&\leq \Bigl\lceil \frac{4-2X(e)-V_{F}(e)-\alpha(2-V_M(e)-V_{F}(e))}{2\alpha -1}-1 \Bigr\rceil\\&= \Bigl\lceil \frac{5-2X(e)-V_{F}(e)-\alpha(4-V_M(e)-V_{F}(e))}{2\alpha-1}\Bigr\rceil =\gamma < |M_e|,
    \end{align*}
which is a contradiction. \end{proof}

\subsection{Analysis of the Local Search algorithm}\label{sec:Analysis}

Let $F$ be an $(\alpha, N_{\max})$-critical set, $\OPT$ be an optimal solution for $\cvca$ and  $M$ be a maximal matching from $\OPT[V']$ (i.e. comprised of links from $OPT$ with both endpoints outside $V(F)$). Let $\ALG(F)$ be the set obtained by adding to $F$ a minimal completion. To ease notation, let $m = |M|/n$, $f = |V(F)|/n$, $\opt = \OPT/n$ and $\alg(F) = |\ALG(F)|/n$. In particular, if $(F,Q)$ is the pair obtained by our Local Search algorithm, then its approximation guarantee is $|\ALG(F)|/|\OPT| = \alg(F)/\opt$. The next lemma provides a first better-than-$2$ bound on this guarantee.

\begin{lemma}\label{lem:firstbound} If $N_{\max}\geq \ell +1$, where $\ell=\lceil \frac{5-2\alpha}{2\alpha-1}\rceil$, then $\opt\geq \max\{1/2, 1-f-m, m+ \frac{m}{\ell}\}$. Moreover, if $\alpha=3/4$, then $\alg(F)/\opt \leq 63/32 = 1.96875$.
\end{lemma}
\begin{proof}
    
    By Lemma~\ref{Every_pair_must_be_crossed}, $\opt\geq 1/2$. To derive $\opt \geq 1 - f - m$, note that there is no link $e\in \OPT$ between two vertices of $V(C_n) \setminus (V(M) \cup V(F))$ because $M$ is a maximal matching. Therefore, each vertex in  $V(C_n) \setminus (V(M) \cup V(F))$ is covered by a different link in $\OPT\setminus M$. We conclude that $$\opt \geq m + \frac{1}{n}\cdot|V(C_n) \setminus (V(M) \cup V(F))| = m + (1 - 2m - f) = 1 - f - m.$$ 

    Since $\OPT$ is feasible, $\OPT \setminus M$ needs to connect all the links of $M$. By Lemma \ref{e_crosses_at_most_k_edges_in_M}, each link in $S \setminus (M \cup F)$ connects at most $\ell$ links of $M$. Then, $\OPT\setminus M$ contains at least $|M|/\ell$ links. Therefore, $\opt \geq m + m/\ell$, or equivalently $m\leq \opt\, \ell/(\ell+1)$. This proves the last bound for $\opt$.

    Now, by Lemma  \ref{bound_for_alg}, $\alg(F) \leq \alpha + (1-f-m)(1-\alpha) + m (1-\alpha)$. Using that $1\le 2\opt$, $(1-f-m)\leq \opt$ and $m\leq \opt\,\ell/(\ell +1)$, we get $\alg(F)/\opt \leq (2\ell + 1+\alpha)/(\ell+1)$. Among the values of $\alpha\in (1/2,1]$ with $(5-2\alpha)/(2\alpha - 1)$ integer, the smallest bound occurs when $\alpha=3/4$, or equivalently $\ell=7$, giving the claimed guarantee (see Figure~\ref{grafico_primera_cota}).
\end{proof}

    \begin{figure}[t]
	\begin{center}
		\includegraphics[scale=.67]{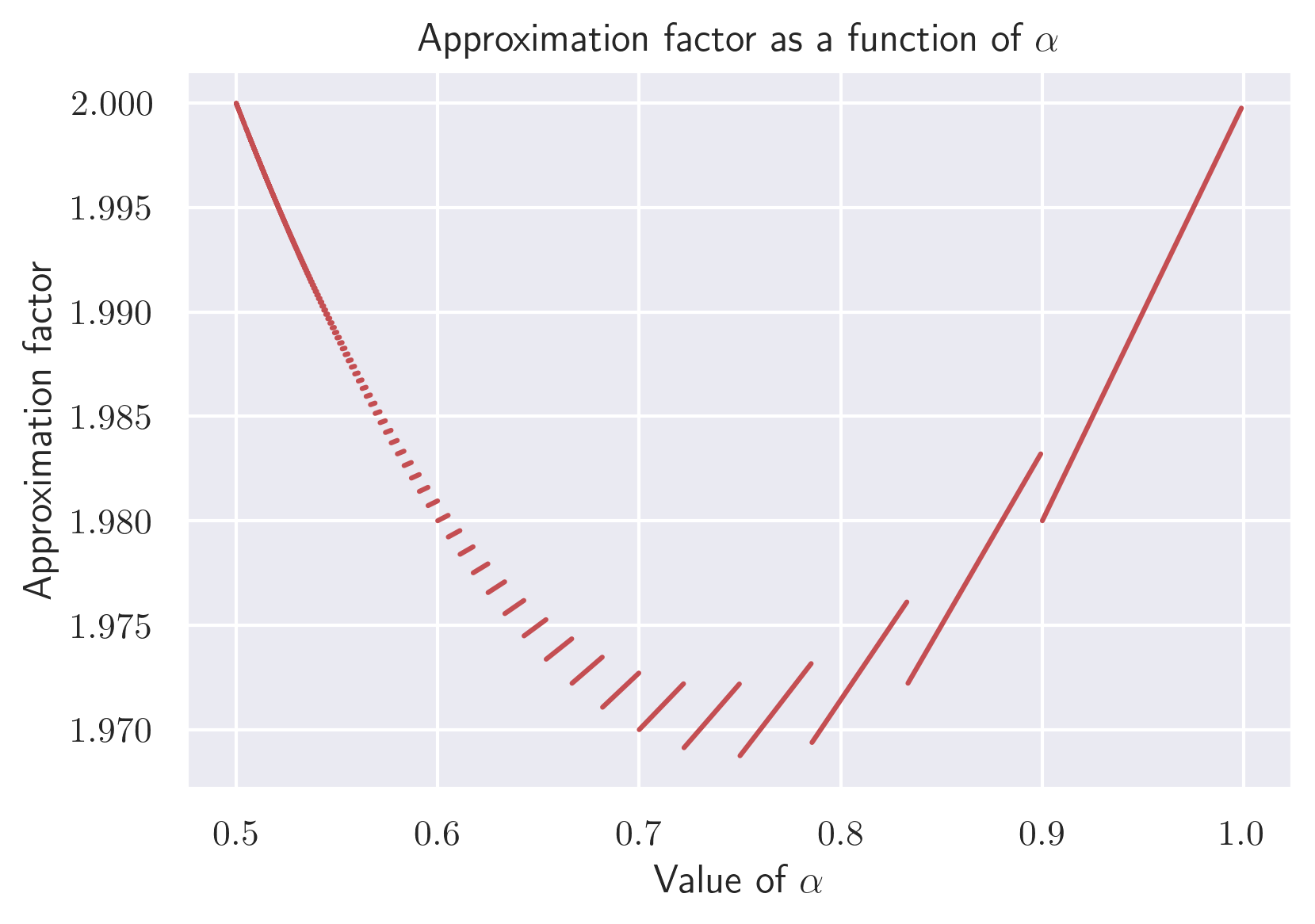}
	\end{center}
	\caption{Plot of the function $f(\alpha)=(2\ell + 1+\alpha)/(\ell+1)$ from the proof of Lemma~\ref{lem:firstbound}. We can observe that the function is minimized when $\alpha=3/4$.}
	\label{grafico_primera_cota}
\end{figure}

\section{Improving the algorithm}\label{sec:improving}

In the previous section we proved that our Local Search algorithm is a $(63/32)$-approximation. We will now describe how to refine the analysis and enhance the algorithm so as to obtain the claimed $1.8704$-approximation.

\subsection{Refining the analysis}\label{refined_analysis}

We start by building a linear programming formulation that provides a stronger lower bound for $|\OPT|$. To do so, we partition $\OPT$ into several subsets (one of them will be $M$), assigning a variable to each one of them representing its size, and provide relations and constraints that these sizes should satisfy.

Let $A=A(F)$ and $B=B(F)$ be the internal and border vertices in $V(F)$ respectively. Let $C=V(M)$ and $D=V(C_n)\setminus (A\cup B\cup C)$ so that $A, B, C$ and $D$ partition $[n]$. Let $\OPT_{ij}\subseteq \OPT\setminus M$, where $i,j\in \{A,B,C,D\}$, be the links with an endpoint in $i$ and the other in $j$. Furthermore, denote by $\OPT_{ij}^k$, where $k\in \{M,P,MP,R\}$, the set of links in $\OPT_{ij}$ that cross $M$ (but do not cross any chord for the perimeter $P$), $P$ (but do not cross $M$), both ($MP$) or none $(R)$ respectively. We get then the partition $\OPT = M \cup \bigcup_{ij,k} \OPT_{ij}^k$. 
The variable $x_{ij}^k$ will represent $|\OPT_{ij}^k|/n$. We will also define variables $x_A$, $x_B$ and $x_M$, which denote the number of vertices in $A$ and $B$ and the number of links in $M$ respectively. Let also $\ell_{ij}^k$, for $k\in \{M,MP\}$, be the upper bound (given by Lemmas \ref{links_can_cross_at_most_one_chord} and \ref{e_crosses_at_most_k_edges_in_M}) on the number of links of $M$ that a link $e\in\OPT_{ij}^k$ could cross. These values only depend on $\alpha$ (and not on $F$). We can then derive the linear program LP$(n,|V(F)|,\alpha)$ that is described in Figure \ref{fig:LP}. Notice that $n$, $|V(F)|$ and all $\ell_{ij}^k=\ell_{ij}^k(\alpha)$ are the only constants in the LP, so its value only depends on $|V(F)|$ and $\alpha$.

\begin{figure}[t]
    \centering
    \begin{minipage}{.53\textwidth}
        \centering
        \begin{align*}
(\text{LP})\qquad\quad  \min x_M + &\textstyle \sum_{ij,k} x_{ij}^k \\
\textstyle \sum_{k}(2x_{aa}^k + x_{ab}^k+x_{ac}^k+x_{ad}^k) & \geq x_A\\
\textstyle\sum_{k}(x_{ab}^k + 2x_{bb}^k+x_{bc}^k+x_{bd}^k) & \geq x_B\\
\textstyle\sum_{k}(x_{ad}^k + x_{bd}^k+x_{cd}^k +2x_{dd}^k)& \geq n-2x_M-x_A-x_B\\
\textstyle\sum_{ij}(x_{ij}^{P} + x_{ij}^{MP}) & \geq x_B/2\\
\textstyle\sum_{ij}(\ell_{ij}^Mx_{ij}^{M} + \ell_{ij}^{MP}x_{ij}^{MP}) & \geq x_M\\
x_A+x_B&\geq|V(F)|\\
2x_M+x_A+x_B&\leq n\\
x&\geq 0\\\
x_{ik}^k&=0 \text{ for $(ij,k)\in \mathcal{Z}$}.
\end{align*}
    
    \end{minipage}%
    \qquad 
    \begin{minipage}{0.4\textwidth}
\begin{align*}
    \ell_{ac}^{MP}&= \lceil (2-2\alpha)/(2\alpha -1)\rceil\\
    \ell_{ad}^{MP}&= \lceil (2-3\alpha)/(2\alpha -1)\rceil_+\\
    \ell_{bb}^M&= \lceil (3-2\alpha)/(2\alpha -1)\rceil\\
    \ell_{bc}^M&=\lceil (4-2\alpha)/(2\alpha -1)\rceil\\
    \ell_{bc}^{MP}&=\lceil (2-2\alpha)/(2\alpha -1)\rceil\\
    \ell_{bd}^M&=\lceil (4-3\alpha)/(2\alpha -1)\rceil\\
    \ell_{bd}^{MP}&=\lceil (2-3\alpha)/(2\alpha -1)\rceil_+\\
    \ell_{cc}^M&= \lceil (5-2\alpha)/(2\alpha -1)\rceil\\
    \ell_{cd}^M&=\lceil (5-3\alpha)/(2\alpha -1)\rceil.
\end{align*}
    \end{minipage}
\vspace{4pt}
The rest of $\ell_{ij}^k$ are equal to $0$ and $\mathcal{Z}=\{(aa,M)$, $(aa,MP)$, $(aa,P)$, $(ab,M)$, $(ab,MP)$, $(ac,M)$, $(ac,R)$, $(ad,M)$, $(ad,R)$, $(bb,MP)$, $(cc,P)$, $(cc,MP)$, $(cd,P)$, $(cd,MP)$, $(dd,M)$, $(dd,P)$, $(dd,MP)$, $(dd,R)\}$.
\caption{Linear program LP$(n,|V(F)|,\alpha)$}\label{fig:LP}
\end{figure}

A careful analysis of the optimal solution for this formulation provides the following improved bound, please refer to the Appendix for the details. In what follows, we define $$f_{\alpha}:=3/(6+4\lceil(4-3\alpha)/(2\alpha-1)\rceil+2\lceil(2-3\alpha)/(2\alpha-1)\rceil_+).$$

\begin{lemma} \label{lem:approximation_order_for_local_search_algorithm} 
Let $F$ be an $(\alpha,N_{\max{}})$-critical set with $N_{\max{}}\geq \lceil \frac{5-2\alpha}{2\alpha-1}\rceil + 1$. 
Then, it holds that $\alg(F)/\opt \leq 2-2(1-\alpha)f_\alpha.$
\end{lemma}

We can now prove a first improvement on the approximation ratio of the Local Search algorithm.

\begin{theorem} The approximation ratio of the Local Search Algorithm (Algorithm~\ref{LocalSearchAlgorithm}) is at most $233/121\approx 1.92562$. \end{theorem}

\begin{proof} 
The expression $f_{\alpha}$ is an increasing step function of $\alpha$ with jumps and breakpoints on the countably infinite set $\mathcal{A}$ in which $(2 - 3\alpha)/(2\alpha-1)$ or  $(4-3\alpha)/(2\alpha-1)$ are integers. Thus, $(1-\alpha)f_{\alpha}$ is the area of the maximal rectangle with base $[\alpha, 1]$ lying below the plot of $f_{\alpha}$, and our guarantee is $2$ minus twice this area, which is maximized on the point $\alpha^*=8/11 \in \mathcal{A}$ (see Figure~\ref{fig:refined-plot}).
\end{proof}

\subsection{A Refined Local Search algorithm}\label{sec:Refinement}

The first phase of Algorithm~\ref{LocalSearchAlgorithm} stops when $F$ is $(\alpha,N_{\max{}})$-critical. 
By increasing $\alpha$ a bit at the end of the phase we may be able to find a set $F'$ with larger utility than $F$, achieving a better solution at the end. We can exploit this intuition and obtain a refined Local Search algorithm. 

\begin{algorithm}[h]
	\caption{Refined Local Search algorithm}\label{CombinedAlgorithm}
	\hspace*{\algorithmicindent} \textbf{Input:} 
	Instance $(C_n,S)$ of $\cvca$. $N_{\max} \in \mathbb{N}$. Array of values $A=\{\alpha_1,\alpha_2,\cdots,\alpha_s, \alpha_{s+1}\}$, where $1/2 < \alpha_1 < \alpha_2 < \cdots < \alpha_s \leq \alpha_{s+1}=1$,
	\begin{algorithmic}[1]
	\State $F_0\gets \emptyset$.
	\For {each $i\in [s]$}
		\State Starting from the current set $F_{i-1}$, run iterations of the first phase of the Local Search algorithm with $\alpha=\alpha_i$ to find an $(\alpha_i,N_{\max{}})$-critical set $F_i$ containing $F_{i-1}$.
	\EndFor 
		\State Find a minimal completion $Q$ of $F_s$
		\State Return $(Q,F_s)$.
	\end{algorithmic}
\end{algorithm}

This allows to prove a generalized version of Lemma~\ref{lem:approximation_order_for_local_search_algorithm}.

\begin{lemma} \label{lem:multiplealfa} Let $\ALG$ be the output of the Refined Local Search algorithm for $N_{\max{}}\geq \lceil \frac{5-2\alpha_1}{2\alpha_1-1}\rceil + 1$, and let $\varepsilon>0$. It holds that 
$|\ALG|/|\OPT| \leq 2 - 2\sum_{i=2}^{m+1}(\alpha_{i} -\alpha_{i-1})f_{\alpha_j}$.
\end{lemma}

\begin{proof} Consider the sets $F_0,F_1,\dots, F_s$ created by the Refined Local Search algorithm and denote $f^i = |F_i|/n$. For every $i\in [s]$, $F_{i}$ was obtained from $F_{i-1}$ by iteratively adding sets such that the marginal utility of that set is at least $(1-\alpha_i)$ times the number of newly added vertices. We conclude that $U(F_i)-U(F_{i-1})\geq (1-\alpha_i)(|V(F_i)|-|V(F_{i-1})|)$. Using that $U(F_0)=0=|V(F_0)|$ and $\alpha_s=1$, we conclude that 
$U(F_s)\geq \sum_{j=1}^s(1-\alpha_j)(|V(F_s)|-|V(F_{i-1}|)=n\sum_{j=1}^{s}(\alpha_{j+1}-\alpha_{j})f^j$. 

Using that  $3n=3n\alpha_1 + 3n\sum_{j=1}^s(\alpha_{j+1}-\alpha_j)$ and Lemma~\ref{bound_alg}, we get that
\begin{align*}
\frac{|\ALG|}{|\OPT|}\leq \frac{-2n+(3n-U(F_s))}{n\,\opt} \leq 2(-2+3\alpha_1)+\sum_{j=1}^s(\alpha_{j+1}-\alpha_j)\frac{(3-f^j)}{\opt}.
\end{align*}
Note that if $f^j\geq f_{\alpha_j}$, then $(3-f^j)/\opt \leq (3-f_{\alpha_j})/\opt \leq 6-2f_{\alpha_j}$.
On the other hand, if $f^j\leq f_{\alpha_j}$, then as $f^j=|V(F_j)|/n$ and $F_j$ is $(\alpha_j,N_{\max{}})$-critical, it is possible to prove that $(3-f^j)/\opt \leq 6-2f_{\alpha_j}$ (see Lemma~\ref{lema:lpvalue} for the details).

Using these bounds on the expression for $|\ALG|/|\OPT|$ and that $6=6\alpha +6\sum_{j=1}^s(\alpha_{j+1}-\alpha_j)$ we get,
\begin{align*}
\frac{|\ALG|}{|\OPT|}\leq (6\alpha_1-4)+\sum_{j=1}^s(\alpha_{j+1}-\alpha_j)(6-2f_{\alpha_j})= 2 - 2\sum_{j=1}^s(\alpha_{j+1}-\alpha_j)f_{\alpha_j}.
\end{align*} \end{proof}

Now we can conclude our main theorem.

\begin{proof}[Proof of Theorem~\ref{thm:local_search}] Recall that $f_{\alpha}=3/(6+4\lceil(4-3\alpha)/(2\alpha-1)\rceil+2\lceil(2-3\alpha)/(2\alpha-1)\rceil_+)$ is an increasing step function of $\alpha$ with jumps on the countably infinite set $\mathcal{A}$ in which $(2 - 3\alpha)/(2\alpha-1)$ or  $(4-3\alpha)/(2\alpha-1)$ are integers. The function $\sum_{j=1}^s(\alpha_{j+1}-\alpha_j)f_{\alpha_j}$ is thus the area of the lower approximation by rectangles whose base are all the  intervals $[\alpha_{j}, \alpha_{j+1}]$, and attains its maximum value $\int_{1/2}^1f_{\alpha} d\alpha$ by setting the values $\alpha_i$ as the sequence $\mathcal{A}$ of breakpoints of $f_{\alpha}$ (see Figure~\ref{fig:refined-plot}). However, we require to only use finitely many points from $\mathcal{A}$; we cannot use points arbitrarily close to $1/2$ as this would require $N_{\max}$ to be arbitrarily large. By properly truncating the sequence, we can achieve in time $n^{O(1/\varepsilon)}$ a guarantee close to $2-2\int_{1/2}^1f_{\alpha} d\alpha \le 1.87032$.

\begin{figure}
	\begin{center}
		\resizebox{\textwidth}{!}{\includegraphics[scale=1.0]{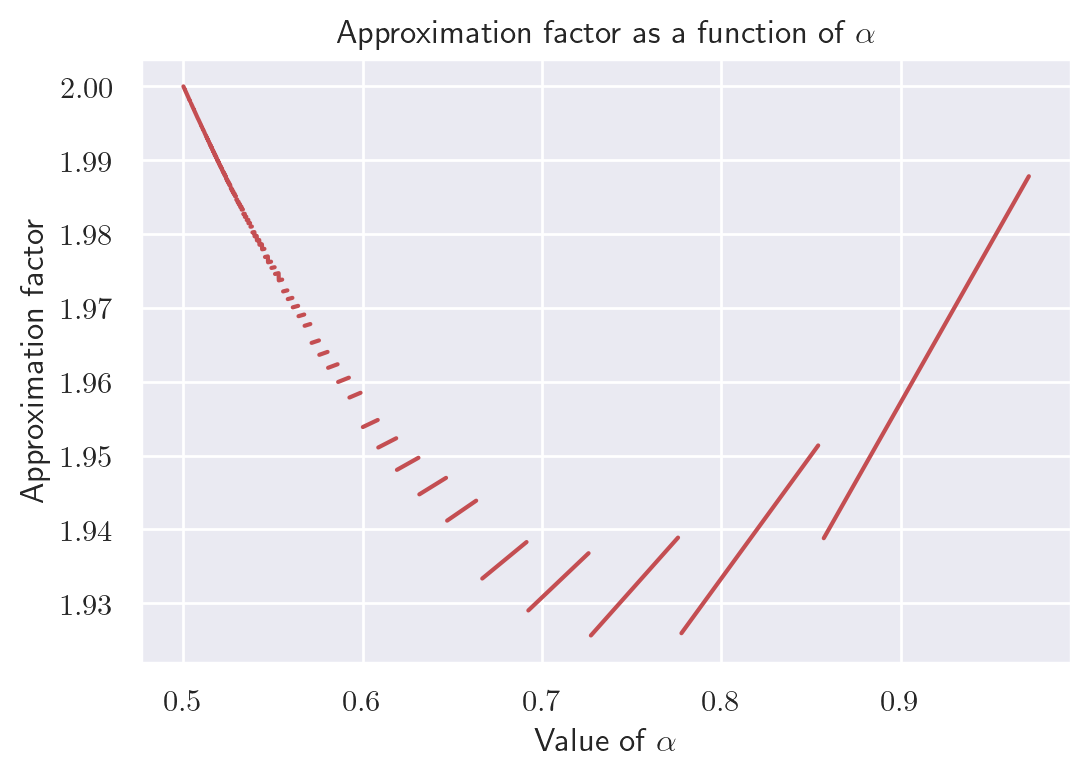}\hspace{15pt}\includegraphics[scale=1.0]{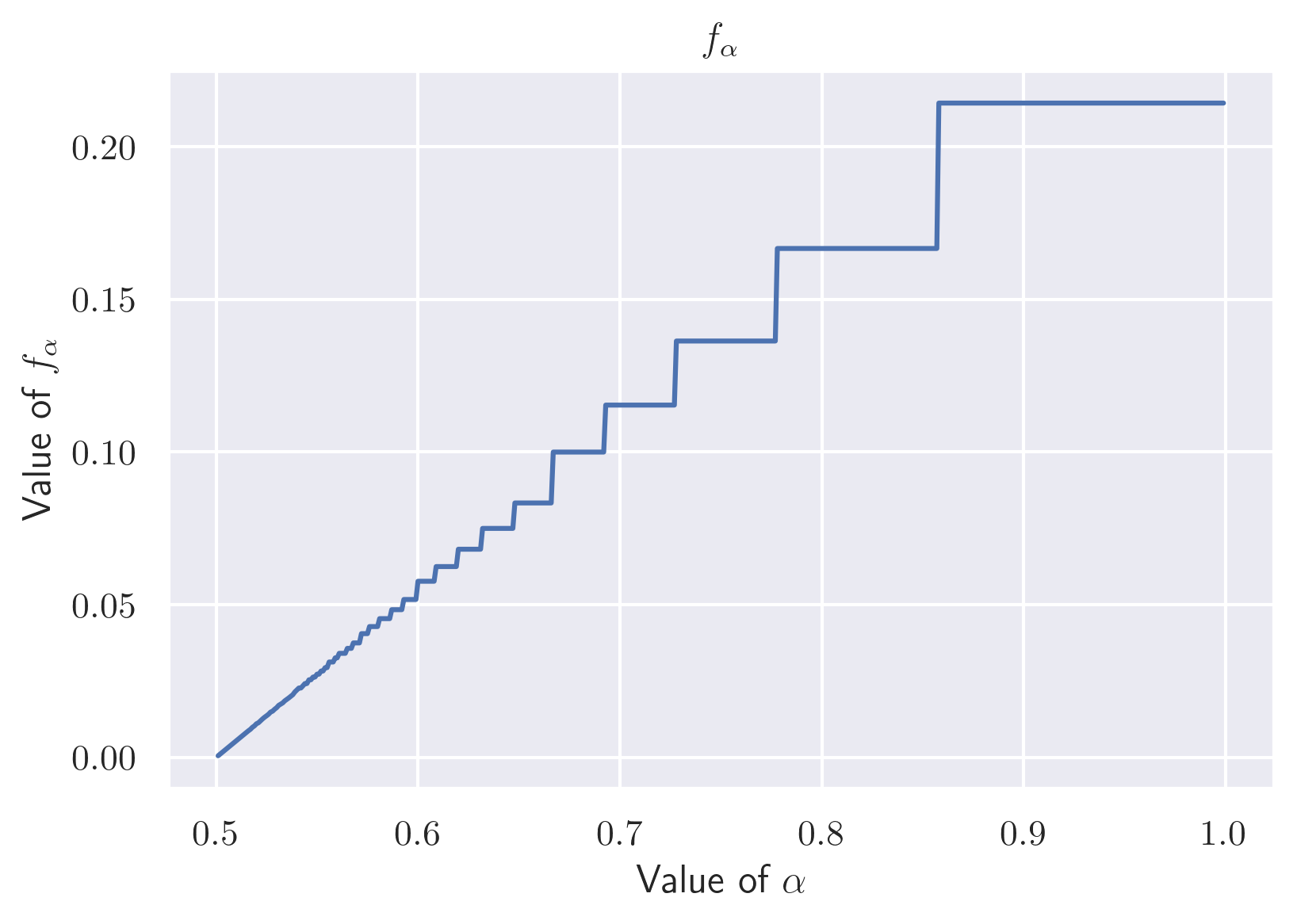}}		
	\end{center}
	\caption{\textbf{Left:} Plot of the function $f(\alpha)=2 - 2(1-\alpha)f_\alpha$ from Lemma~\ref{lem:approximation_order_for_local_search_algorithm}. We can observe that this function is minimized when $\alpha=8/11$.
	\textbf{Right:} Plot of $f_\alpha$ that illustrates the area under the curve that must be computed to derive the final approximation ratio according to the proof of Theorem~\ref{thm:local_search}.}
	\label{fig:refined-plot}
\end{figure}

More in detail, we truncate the summation of Lemma \ref{lem:multiplealfa} to the $k$-th term, obtaining the corresponding value of $\alpha$ by solving $\frac{4-3\alpha}{2\alpha -1} = k$, thus being $\alpha_k = \frac{4+k}{2k+3}$. We rewrite $\alpha_k$ as $1/2 + \hat{\varepsilon}$, obtaining $\hat{\varepsilon}=\frac{5}{4k + 6}$.
As $f_\alpha$ is an increasing function, we can bound the obtained approximation ratio by $$2 - 2\int_{1/2+\hat{\varepsilon}}^1f_{\alpha} d\alpha \leq 2 - 2\int_{1/2}^1f_{\alpha} d\alpha + 2\int_{1/2}^{1/2 + \hat{\varepsilon}} f_{\alpha_k} d\alpha \le 1.87032 + 2\hat{\varepsilon} f_{\alpha_k}.$$ The value $k$ can be chosen large enough so that the term $2\hat{\varepsilon} f_{\alpha_k}$ is small enough, leading to the claimed bound.

Regarding the running time, notice that every iteration of the first phase takes $|S|^{O(N_{\max})}$, but $|S|=O(n^2)$ and, for $\alpha=1/2 + \varepsilon$, $N_{\max} \geq \lceil 2 \frac{1}{\varepsilon} - 1 \rceil = O(1/\varepsilon)$. Moreover, we have to repeat this iteration $\lceil(4-3(\frac{1}{2} + \varepsilon)/(2(\frac{1}{2} + \varepsilon)-1)\rceil = O(1/\varepsilon)$ times. Therefore, the total running time for the algorithm is $O(1/\varepsilon) n^{O(1/\varepsilon)} = n^{O(1/\varepsilon)}$.\end{proof}

\section{Hardness of Approximation of \cvca} \label{hardness_of_approximation}

In this section we prove that \cvca is APX-hard. The following theorem summarizes a hardness result from Gálvez et al.~\cite{Galvez19} in the context of edge-connectivity augmentation, stated in a way that will be useful for our purposes.

\begin{theorem}{\cite{Galvez19}}\label{thm:apxhard-edge}
There exists a constant $\delta_0>0$ such that it is NP-hard to decide whether an instance of Cycle edge-connectivity augmentation admits a solution of size $\frac{n}{2}$, or no feasible solution for the instance has size at most $(1+\delta_0)\frac{n}{2}$.
\end{theorem}

Using this we can prove that finding arbitrarily good solutions for our problem in polynomial time is hard as well.

\begin{theorem}\label{thm:apxhard-node}
There exists a constant $\delta>0$ such that there is no $(1+\delta)$-approximation for \cvca unless P $=$ NP.
\end{theorem}

\begin{proof}
    We will show that if $(C_n,S)$ is an instance of Cycle edge-connectivity augmentation, then \begin{enumerate} \item If $(C_n,S)$ admits a solution of size $n/2$, then the same input interpreted as an instance of \cvca admits a feasible solution of size $\frac{n}{2}$ too; and
		\item If $(C_n,S)$ does not admit a feasible solution of size at most $(1-\delta_0)\frac{n}{2}$, then neither does the same input as an instance of \cvca. \end{enumerate}
	
	The second statement follows directly since $3$-vertex-connectivity implies $3$-edge connectivity. For the first statement, we will prove that a feasible solution of size $n/2$ for Cycle edge-connectivity augmentation actually consists of a single circle component, being then feasible for \cvca thanks to Lemma~\ref{circle-component-characterization}. Indeed, if the solution has more than one circle component, no pair of components share border nodes as every node of the cycle is incident to exactly one link, and furthermore one of them has to be contiguous (i.e a circle component $L$ such that $V(L)$ is an interval on the cycle). However, in this case it is possible to disconnect the latter component by deleting the two edges of the cycle incident to its border nodes. \end{proof}

\newpage
 \bibliography{bibliography.bib}

\appendix

\section{An LP-based lower bound for the optimal number of links}\label{app:refined_analysis}

In this section we provide a detailed analysis of the Linear Programming formulation provided in Section~\ref{sec:improving} (see Figure~\ref{fig:LP}). First of all, it is possible to prove that $\OPT_{ij}^k=\emptyset$ for all $(ij,k)\in \mathcal{Z}$ and prove that the values $\ell_{ij}^k$ from the LP are valid bounds, as the following lemma states.

\begin{lemma}\label{zeros-LP}
Let
$\mathcal{Z}=\{(aa,M)$, $(aa,MP)$, $(aa,P)$, $(ab,M)$, $(ab,MP)$, $(ac,M)$, $(ac,R)$, $(ad,M)$, $(ad,R)$, $(bb,MP)$, $(cc,P)$, $(cc,MP)$, $(cd,P)$, $(cd,MP)$, $(dd,M)$, $(dd,P)$, $(dd,MP)$, $(dd,R)\}$. It holds that $\OPT_{ij}^k=\emptyset$ for all $(ij,k)\in \mathcal(Z)$.

Also, if $k \in \{M,MP\}$, there exist values $\ell_{ij}^k$ such that a link in $\OPT_{ij}^k$ can only cross at most $\ell_{ij}^k$ links, and these values can be computed explicitly. 

Furthermore, variables $x_{ij}^R$ (for all $ij\neq aa$), $x_{ac}^P,x_{bc}^P$, and $x_{bd}^P$ can be set to zero in any optimal solution of the linear program.

\end{lemma}

In order to prove this lemma we require the following technical results, which are also useful in subsequent proofs.

\begin{lemma}\label{technical}
Let $e=ab$ be a chord in $C_n$ and $\mathcal{L}_e$ be a collection of non-crossing circle components, such that all of them cross $e$. Then $$|V(\bigcup_{J\in \mathcal{L}_e}J)|\geq |\mathcal{L}_e|+2+\sum_{J\in \mathcal{L}_e}(|V(J)|-3).$$
\end{lemma}

\begin{proof}
The result is straightforward if $|\mathcal{L}_e|=1$. So let us assume that the collection $\mathcal{L}_e$ contains $k\geq 2$ circle components. Since all the vertices of a component $L\in \mathcal{L}_e$ (and thus, all the internal and border chords of $L$) are completely contained in a zone of every other component $K\in \mathcal{L}_e\setminus \{L\}$ (see Section~\ref{sec:minimal_comp} for a definition of the zone graph and related concepts), we know that $e$ cannot be an internal or a border chord of any $L\in \mathcal{L}_e$. By Lemma \ref{general_crossing} (see Section~\ref{sec:circle_comp}), $e$ must then cross a border chord of every component. In fact, if both extremes of $e$ are outside $V(L)$ then $e$ must cross exactly two border chords $f_L$ and $g_L$ of $L$, and if one extreme is in $V(L)$ and the other is outside, then $e$ only crosses one border chord $f_L$ of $L$. Abusing notation, let $g_L=f_L$ in this case.

Note that if we draw $C_n$ in the plane as a circle, and we draw the links as straight line segments, then $e$ crosses all the border chords in the multiset $\{f_L,g_L\}_{L\in \mathcal{L}_e}$ in certain order. In fact, one can enumerate the circle-components of $\mathcal{L}_e$ as $L_1, \dots, L_k$ in such a way that $e$ crosses the chords described above in the order $f_{L_1},g_{L_1}, f_{L_2}, g_{L_2}, \dots, f_{L_k}, g_{L_k}$, where we may have that for some $i$, $g_{L_i}$ is the same as $f_{L_{i+1}}$. 
    (To be more formal, w.l.o.g. we can assume that one of the endpoints of $e$ is the vertex $n$, and the other some vertex $j\in [n-1]$. Then all border chords in the list above can be seen as intervals in $[1, n-1]$ that contains $j$ and that does not cross each other. Therefore, they can be sorted by inclusion, yielding the order above. The fact that both border chords of a component appear before both border chords of the next one follows from the fact that every circle component is fully contained in a closed zone of every other in the list.)     For every $i\in [k-1]$, all vertices in $V(L_1\cup \dots \cup L_i)$ are in the closed zone defined by the border chord $f_{L_{i+1}}$ of $L_{i+1}$. From here we conclude that $V(L_1\cup \dots \cup L_i)$ and $V(L_{i+1})$ may intersect in at most 2 vertices. Therefore, $|V(L_1\cup \dots \cup L_{i+1})|\geq |V(L_1\cup \dots \cup L_{i})|+|V(L_{i+1})|-2$. Iterating this argument, we obtain  $|V(L_1\cup \dots \cup V_k)|\geq 2+ \sum_{i=1}^k(|V(L_{i})|-2)= k+2 + \sum_{i=1}^k(|V(L_{i})|-3)$ as needed.   \end{proof}
    
\begin{lemma}\label{atmost1}
Let $e$ be a chord in $C_n$, then $e$ crosses at most 1 circle component of $\mathcal{L}$. 
\end{lemma}

\begin{proof}
Let $\mathcal{L}_e$ be the circle components in $\mathcal{L}$ that cross $e$, $X=\bigcup_{J\in \mathcal{L}_e}J$ so that $X\cup \{e\}$ is the unique circle-component created by adjoining $e$ to $F$. Note that if $e$ crosses at least one circle component then $F\cup \{e\}$ is singleton-free. 

Suppose first that $|\mathcal{L}_e|\geq 3$. By Lemma \ref{technical} we have that $|V(X\cup \{e\})|\geq |V(X)|\geq  5+\sum_{J\in \mathcal{L}_e}(|V(J)|-3)$. Therefore,
    \begin{align*}
    U(F\cup \{e\}) &= -|F\cup \{e\}| + (|V(X\cup\{e\})|-3) +\sum_{J \in \mathcal{L}\setminus \mathcal{L}_e}(|V(J)|-3)\\
&\geq -(|F| + 1) + 2+\sum_{J\in \mathcal{L}}(|V(J)|-3))\geq  U(F)+1.
\end{align*}
Since $e$ connects at most 2 new vertices, $U(F\cup \{e\})-U(F) \geq 1 \geq |V(F\cup \{e\})\setminus V(F)|/2 \geq (1-\alpha)|V(F\cup \{e\}\setminus V(F)|$, which is not possible since $F$ is $(\alpha,N_{\max})$-critical and $N_{\max}\geq 1$.

Suppose now that $|\mathcal{L}_e|=2$, say $\mathcal{L}_e=\{L_1, L_2\}$. Since $L_1$ and $L_2$ do not cross, $|V(L_1)\cap V(L_2)|\leq 2$ and so $|V(L_1)|+|V(L_2)|-2\leq |V(L_1\cup L_2)|$. Then,
\begin{align*}
U(F\cup \{e\})-U(F)&=-1+(|V(Y)|-3)-(|V(L_1)|-3)-(|V(L_2)|-3)\\
&=2 + |V(\{e\}\cup L_1 \cup L_2)|-|V(L_1)|-|V(L_2)|\\
&\geq |V(\{e\}\cup L_1\cup L_2)|-|V(L_1\cup L_2)|. 
\end{align*}

If $e$ has at least one endpoint in $V'=V(C_n)\setminus V(F)$, then the right hand side of the equation above is at least 1, and thus $U(F\cup \{e\})-U(F) \geq  1\geq |V(F\cup \{e\})\setminus V(F)|/2 \geq (1-\alpha)|V(F\cup \{e\}\setminus V(F)|$, which is again not possible.

On the other hand, if $|V(F\cup \{e\})\setminus V(F)|=0$, we have
$U(F\cup \{e\})-U(F)\geq |V(\{e\}\cup L_1\cup L_2)|-|V(L_1\cup L_2)| \geq 0 = (1-\alpha)|V(F\cup \{e\}\setminus V(F)|$, which is also not possible.
\end{proof}

\begin{proof}[Proof of Lemma~\ref{zeros-LP}] We will do this in two steps, first we will analyze all the possible combinations of $ij$, where $i$, $j \in \{A,B,C,D\}$.

\paragraph{Step 1: Exhaustive analysis of the variables of the Linear Program}
    
\paragraph{Case 1: \texorpdfstring{$ij=aa$}{ij=aa}} If a link $e$ with both extremes in $A$ crosses some $m \in M$ or some $p\in P$, then both endpoints of $e$ are in different components and $e$ crosses both. This cannot happen due to Lemma \ref{atmost1}. Therefore $\OPT_{aa}^M=\OPT_{aa}^P=\OPT_{aa}^{MP}=\emptyset$.

\paragraph{Case 2: \texorpdfstring{$ij=ab$}{ij=ab}} If a link $e$ with extremes in $A$ and $B$ crosses some link of $M$, then we deduce that both endpoints are in different components and $e$ crosses exactly one of them (it cannot
cross two due to Lemma \ref{atmost1}). But then $e$ crosses some $p\in P$. We conclude that $\OPT_{ij}^M=\emptyset$. But $e$ cannot cross both a link of $M$ and one of $P$, by Lemma \ref{e_crosses_at_most_k_edges_in_M} with $X(e)=1$, $V_F(e)=2$ and $V_M(e)=0$, we get that $e$ should connect no more than $\bigl\lceil \frac{1-2\alpha}{2\alpha -1}\bigr\rceil < 0$ links of $M$. Therefore, $\OPT_{ij}^{MP}=\emptyset$.

\paragraph{Case 3: \texorpdfstring{$ij=ac$}{ij=ac}} Any link with extremes in $A$ and $C$ must cross $P$, therefore $\OPT_{ac}^M=\OPT_{ac}^R=0$. Using Lemma \ref{e_crosses_at_most_k_edges_in_M} with $X(e)=1$, $V_F(e)=1$ and $V_M(e)=1$, we get that any link in $\OPT_{ac}^{MP}$ connects at most $\ell_{ac}^{MP}:=\bigl\lceil \frac{2-2\alpha}{2\alpha -1}\bigr\rceil$ links of $M$.

\paragraph{Case 4: \texorpdfstring{$ij=ad$}{ij=ad}} Any link $e$ with extremes in $A$ and $D$ must cross $P$, 
therefore $\OPT_{ac}^M=\OPT_{ac}^R=0$. Using Lemma \ref{e_crosses_at_most_k_edges_in_M} with $X(e)=1$, $V_F(e)=1$ and $V_M(e)=0$, we get that any link in $\OPT_{ac}^{MP}$ connects at most $\ell_{ad}^{MP}=\bigl\lceil \frac{2-3\alpha}{2\alpha -1}\bigr\rceil_+$ links of $M$.

\paragraph{Case 5: \texorpdfstring{$ij=bb$}{ij=bb}} Any link $e$ with both extremes in $B$ cannot cross $M$ and $P$ at the same time, since by Lemma \ref{e_crosses_at_most_k_edges_in_M} with $X(e)=1$, $V_F(e)=2$ and $V_M(e)=0$, any such link would connect at most $\bigl\lceil \frac{1-2\alpha}{2\alpha -1}\bigr\rceil< 0$ links of $M$. A link $e$ that only crosses $M$, connects at most $\ell_{bb}^{M}=\bigl\lceil \frac{3-2\alpha}{2\alpha -1}\bigr\rceil$ links of $M$, using lemma \ref{e_crosses_at_most_k_edges_in_M} with $X(e)=0$, $V_F(e)=2$ and $V_M(e)=0$.

\paragraph{Case 6: \texorpdfstring{$ij=bc$}{ij=bc}} Using Lemma  \ref{e_crosses_at_most_k_edges_in_M}, with $X(e)=0$, $V_F(e)=1$ and $V_M(e)=1$, we get that links with extremes in $B$ and $C$ that do not cross $P$ connect at most $\ell_{bc}^{M}=\bigl\lceil \frac{4-2\alpha}{2\alpha -1}\bigr\rceil$ links of $M$. And using $X(e)=1$, $V_F(e)=1$ and $V_M(e)=1$ we note that links with extremes in $B$ and $C$ that do cross $P$, connect at most $\ell_{bc}^{MP}=\bigl\lceil \frac{2-2\alpha}{2\alpha -1}\bigr\rceil$ links of $M$.

\paragraph{Case 7: \texorpdfstring{$ij=bd$}{ij=bd}} %No set is necessarily empty. 
Using Lemma~\ref{e_crosses_at_most_k_edges_in_M} with $X(e)=0$, $V_F(e)=1$ and $V_M(e)=0$, we  get that links with extremes in $B$ and $D$ that do not cross $P$, connect at most $\ell_{bd}^{M}=\bigl\lceil \frac{4-3\alpha}{2\alpha -1}\bigr\rceil$ links of $M$. And using  $X(e)=1$, $V_F(e)=1$ and $V_M(e)=0$ we get  that links with extremes in $B$ and $D$ that do cross $P$, connect at most $\ell_{bd}^{MP}=\bigl\lceil \frac{2-3\alpha}{2\alpha -1}\bigr\rceil$ links of $M$.

\paragraph{Case 8: \texorpdfstring{$ij=cc$}{ij=cc}} A link with both extremes in $C$ cannot cross $P$ due to Lemma \ref{atmost1}. Then $\OPT_{cc}^P=\OPT_{cc}^{MP}=\emptyset$. We also note, using Lemma \ref{e_crosses_at_most_k_edges_in_M} with $X(e)=0$, $V_F(e)=0$ and $V_M(e)=2$, that these links can connect at most  $\ell_{cc}^{M}=\bigl\lceil \frac{5-2\alpha}{2\alpha -1}\bigr\rceil$ links of $M$. 

\paragraph{Case 9: \texorpdfstring{$ij=cd$}{ij=cd}} A link with an extreme in $C$ and other in $D$ cannot cross $P$ due to Lemma \ref{atmost1}; thus $\OPT_{cd}^{P}=\OPT_{cd}^{MP}=\emptyset$. We also note, using Lemma \ref{e_crosses_at_most_k_edges_in_M} with $X(e)=0$, $V_F(e)=0$ and $V_M(e)=1$, that these links can connect  at most $\ell_{cd}^{M}=\bigl\lceil \frac{5-3\alpha}{2\alpha -1}\bigr\rceil$ links of $M$.

\paragraph{Case 10: \texorpdfstring{$ij=dd$}{ij=dd}} In this case $\OPT_{dd}=\emptyset$ because if we had such a link we would have added it to the maximal matching $M$.
\medskip

\paragraph{Step 2: Further simplifications.}

The LP in its current form is quite cumbersome, so we will further simplify it by fixing some variables to zero and proving that this can be done without changing the optimal value that can be achieved. We call $x^*, x_A^*, x_B^*, x_M^*$ an optimal solution to the LP.

\begin{itemize}
\item $x_{cd}^{R}=0$: If $x_{cd}^{R*} = \varepsilon > 0$. Then we can define a new solution $(x', x_A', x_B', x_M')$ by setting $x_{cd}^{R'}= x_{cd}^{R*} - \varepsilon$, $x_{cd}^{M'}=x_{cd}^{M*} + \varepsilon$, maintaining feasibility and getting the same objective value.

\item $x_{bd}^{R}=0$: If $x_{bd}^{R*} = \varepsilon > 0$. Then we can define a new solution $(x', x_A', x_B', x_M')$ by setting $x_{bd}^{R'}= x_{bd}^{R*} - \varepsilon$, $x_{bd}^{M'}=x_{bd}^{M*} + \varepsilon$, maintaining feasibility and getting the same objective value.

\item $x_{bc}^{R}=0$: If $x_{bc}^{R*} = \varepsilon > 0$. Then we can define a new solution $(x', x_A', x_B', x_M')$ by setting $x_{bc}^{R'}= x_{bc}^{R*} - \varepsilon$, $x_{bc}^{M'}=x_{bc}^{M*} + \varepsilon$, maintaining feasibility and getting the same objective value.

\item $x_{ab}^{R}=0$: If $x_{ab}^{R*} = \varepsilon > 0$. Then we can define a new solution $(x', x_A', x_B', x_M')$ by setting $x_{ab}^{R'}= x_{bc}^{R*} - \varepsilon$, $x_{aa}^{R*}=x_{aa}^{R*} + \varepsilon/2$, $x_{bb}^{M*}=x_{bb}^{M*} + \varepsilon/4$, $x_{bb}^{R*}=x_{bb}^{R*} + \varepsilon/4$, maintaining feasibility and getting the same objective value.

\item $x_{bb}^{R}=0$: If $x_{bb}^{R*} = \varepsilon > 0$. Then we can define a new solution $(x', x_A', x_B', x_M')$ by setting $x_{bb}^{R'}= x_{bb}^{R*} - \varepsilon$, $x_{bb}^{M'}=x_{bb}^{M*} + \varepsilon$, maintaining feasibility and getting the same objective value.

\item $x_{ac}^{P}=0$: If $x_{ac}^{P*} = \varepsilon > 0$. Then we can define a new solution $(x', x_A', x_B', x_M')$ by setting $x_{ac}^{P'}= x_{ac}^{P*} - \varepsilon$, $x_{ac}^{MP'}=x_{ac}^{MP*} + \varepsilon$, maintaining feasibility and getting the same objective value.

\item $x_{bc}^{P}=0$: If $x_{bc}^{P*} = \varepsilon > 0$. Then we can define a new solution $(x', x_A', x_B', x_M')$ by setting $x_{bc}^{P'}= x_{bc}^{P*} - \varepsilon$, $x_{bd}^{P'}=x_{bd}^{P*} + \varepsilon$, maintaining feasibility and getting the same objective value.

\item $x_{bd}^{P}=0$:  If $x_{bd}^{P*} = \varepsilon > 0$. Then we can define a new solution $(x', x_A', x_B', x_M')$ by setting $x_{bd}^{P'}= x_{bd}^{P*} - \varepsilon$, $x_{ad}^{P'}=x_{ad}^{P*} + \varepsilon$, $x_A^{'} = x_A^{*} + \varepsilon$, $x_B^{'} = x_B^{*} + \varepsilon$, maintaining feasibility and getting the same objective value.

\item $x_{cc}^{R}=0$: This does not affect as it does not appear in any constraint. 

\end{itemize}
\end{proof}

After all these simplifications, we obtain the following LP.

\begin{align*}
\text{(LP)}\quad  \min (x_M+ x_{aa}^R + x_{ab}^{P} +  x_{bb}^{M} + x_{bb}^{MP} + x_{cc}^{M} &+ x_{ad}^{P} + x_{ad}^{MP} \\
 + x^{MP}_{ac} + x^M_{bc} + x_{bc}^{MP}
+ x_{bd}^M + x_{bd}^{MP} + x_{cd}^M) & \\
2x_{aa}^R  + x_{ab}^{P} +  x_{ad}^{P} + x_{ad}^{MP} + x_{ac}^{MP}  & \geq x_A\\
x_{ab}^{P} + 2(x_{bb}^M + x_{bb}^P)  + x_{bc}^M + x_{bc}^{MP}  + x^M_{bd} + x_{bd}^{MP} & \geq x_B\\
x^P_{ad} + x_{ad}^{MP} + x^M_{bd} + x_{bd}^{MP} + x_{cd}^{M} & \geq n-2x_M-x_A-x_B\\
+ x_{ab}^{P} + x_{ad}^P + x_{ad}^{MP} + x^{MP}_{ac} + x_{bb}^P + x_{bc}^{MP} + x_{bd}^{MP} & \geq x_B/2\\
 \ell_{bb}^M x_{bb}^M + \ell_{ad}^{MP} x_{ad}^{MP} + \ell_{ac}^{MP} x_{ac}^{MP} + \ell_{bd}^M x^M_{bd} 
& \\
+ \ell_{bc}^M x^M_{bc}
+ \ell_{ac}^{MP} x_{bc}^{MP}
+ \ell_{cd}^M x_{cd}^M +
\ell_{ad}^{MP} x_{bd}^{MP} 
& \\
+\ell_{cc}^{M} x_{cc}^M & \geq x_M\\
x_A+x_B&\geq|V(F)|\\
2x_M+x_A+x_B&\leq n\\
x,x_A,x_B,x_M&\geq 0.
\end{align*}

The next observation follows from the fact that the perimeter $P(F)$ is a set of chords covering all the vertices in $B(F)$. 
\begin{proposition}\label{P_is_greater_than_B/2}
$|P(F)|\geq |B(F)|/2$.
\end{proposition}

\begin{lemma}
The assignment $x_{ij}^k=|\OPT_{ij}^k|$, $x_A=|A|$, $x_B=|B|$, $x_M=|M|$ is a feasible solution for LP$(n,|V(F)|,\alpha)$.
\end{lemma}
\begin{proof} The first, second and third constraints hold immediately as $A$, $B$ and $D$ must be covered by links in $\OPT$, and $|D|=n-2|M|-|A|-|B|$. The fourth constraint comes from the fact that the chords of $P$ need to be crossed by some link of $\OPT\setminus M$, but any link of $\OPT\setminus M$ can cross at most one chord in $P$ due to Lemma~\ref{links_can_cross_at_most_one_chord}; hence, the left hand side is at least $|P|$, which, by Proposition ~\ref{P_is_greater_than_B/2}, is at least $|B|/2$. The fifth constraint follows since all the links of $M$ must be connected to the rest of the solution, but the number of links that a link from $S\setminus F$ can cross is bounded thanks to Lemma~\ref{e_crosses_at_most_k_edges_in_M}. The sixth and seventh constraints hold since $|A|+|B|=|V(F)|$ and $2|M|+|A|+|B|=n-|C|\leq n$. %The nonnegativity constraint holds since all the variables represent cardinalities.
\end{proof}

For $1/2<\alpha \leq 1$, and $0\leq x\leq 1$, let \begin{align*}
R(\alpha)&=3+3\ell_{cd}^M(\alpha), & W(x,\alpha)&=(R(\alpha)-S(\alpha)x)/T(\alpha),\\
S(\alpha)&=3+2\ell_{bd}^M(\alpha)+\ell_{ad}^{MP}(\alpha), & f_\alpha&=(2R(\alpha)-T(\alpha))/2S(\alpha)=3/2S(\alpha),\\
T(\alpha)&=3+6\ell_{cd}^M(\alpha), &  \psi_\alpha(x)&=(3-x)/W(x,\alpha).
\end{align*}

Notice that $W(x,\alpha)$ is decreasing in $\alpha$, and $f_\alpha$ is defined so that $W(f_\alpha,\alpha)=1/2$, implying that $\psi_\alpha(f_\alpha)=6-2f_\alpha$. 
We also observe that $3S(\alpha) - R(\alpha) = 6 + 6\ell_{bd}^M + 3\ell_{ad}^{MP} - 3\ell_{cd}^M \geq 6 + 6\frac{4-3\alpha}{2\alpha-1} + 3\frac{2-3\alpha}{2\alpha-1} - 3\frac{6 - 5\alpha}{2\alpha -1} =
12\frac{1 - \alpha}{2\alpha-1}>0$  for all $\alpha > \frac{1}{2}$.
In particular, the function $\psi_\alpha(x)$  satisfies $\psi'_{\alpha}(x)=\frac{T(\alpha)}{(R(\alpha)-S(\alpha)x)^2}(3S(\alpha)-R(\alpha))\geq 0$, so it is an increasing function. This allows to prove the following and to provide improved approximation guarantees.

\begin{lemma} \label{lema:lpvalue} Let $F$ be an $(\alpha,N_{\max{}})$-critical set with $N_{\max{}}>\ell_{cc}^M$. 
Then the optimal value of LP$(n,|V(F)|,\alpha)$ is $\max\{n/2, nW(|V(F)|/n,\alpha)\}$. In particular, this is equal to $n/2$ 
if and only if $|V(F)|/n\geq f_\alpha$.
\end{lemma}

To do this, we study the dual problem of the LP and then we solve both programs using LP solvers, obtaining feasible points that give the same objective value. The latter fact certifies that the point obtained for the primal is indeed an optimal solution of our problem.

The dual problem of the LP is the following

\begin{minipage}{0.5\textwidth}
\begin{align*}
\text{(D)}\quad  \max y_3(n-|V(F)|)+y_6|V(F)|-y_7 n\\
2y_1&\leq 1\\
2y_2 + \ell_{bb}^M y_5 &\leq 1 \\
y_1+y_3+2y_4&\leq 1\\
y_1+y_2+2y_4&\leq 1\\
2y_2+2y_4&\leq 1\\
y_2+y_3+\ell_{w,b}y_5&\leq 1\\
y_2+\ell_{m,b}y_5&\leq 1\\
y_3+\ell_{w,m}y_5&\leq 1\\
\end{align*}
\end{minipage}%
\qquad
\begin{minipage}{0.28\textwidth}
\begin{align*}
\ell_{cc} y_5&\leq 1\\
y_2 + \ell_{bc}^M y_5&\leq 1\\
y_1 + 2y_4 + \ell_{ac}^{MP} y_5&\leq 1\\
y_1 + y_3 + 2y_4 + \ell_{ad}^{MP} y_5&\leq 1\\
y_2 + y_3 + 2y_4 + \ell_{ad}^{MP} y_5&\leq 1\\
2y_3-y_5-2y_7&\leq 1\\
-y_1+y_6-y_7&\leq 0\\
-y_2-y_4+y_6-y_7&\leq 0\\
y_1,y_2,y_3,y_4,y_5,y_6,y_7&\geq 0\\
\end{align*}
\end{minipage}

\begin{table}\label{table1}
\centering
\setlength{\tabcolsep}{0.5em}

\renewcommand{\arraystretch}{1.9}
\begin{tabular}{|c|c|c|c|c|}
\hline
Variable & optimal assignment for $I_1 (\alpha > 2/3)$ & optimal assignment for $I_1 (\alpha \leq 2/3)$ \\
\hline
$x_A$ &  $\displaystyle\frac{1}{3}|V(F)|$ & $\displaystyle \frac{1}{3}|V(F)|$ \\
$x_B$ & $\displaystyle \frac{2}{3}|V(F)|$ & $\displaystyle \frac{2}{3}|V(F)|$ \\
$x_M$ & $x_M^{(1)}$ & $x_M^{(1)}$ \\
$x_{ad}^P$  & $\displaystyle \frac{1}{3}|V(F)|$ & $0$ \\
$x_{ad}^{MP}$  & $0$ & $\displaystyle \frac{1}{3}|V(F)|$ \\
$x_{bd}^M$ & $\displaystyle \frac{2}{3}|V(F)|$ & $\displaystyle \frac{2}{3}|V(F)|$\\
$x_{cd}^M$ & $\displaystyle x_{cd}^{M(1)}$ & $x_{cd}^{M(1)}$ \\
Other & $0$ & $0$ \\
\hline
\end{tabular}
\caption{An optimal solution of the LP as a function of $|V(F)|$, where we denote $X_M^{(1)} = \displaystyle \frac{3\ell_{cd}^{M}n + (\ell_{ad}^{MP} + 2\ell_{bd}^{M} - 6\ell_{cd}^{M})|V(F)|}{(3 +6\ell_{cd}^{M})}$, and $\displaystyle  x_{cd}^{M(1)} =\frac{3n - |V(F)|(2\ell_{ad}^{MP} + 4\ell_{bd}^M+6)}{3+6\ell_{cd}^M} $.}\label{PL_Optimum2}
\end{table} 

We observe that any solution that is primal feasible for (LP) is also feasible when we decrease the value of $|V(F)|$. From here we get that the feasible set for (LP) increases as $|V(F)|$ increases, therefore the optimal value of (LP) must be a decreasing function of $|V(F)|$.

The behavior of LP depends on whether $|V(F)|$ is in $\displaystyle I_1 = \left[0, \frac{3n}{2\ell_{ad}^{MP} + 4\ell_{bd}^M + 6}\right)$ or in $\displaystyle I_2=\left[\frac{3n}{2\ell_{ad}^{MP} + 4\ell_{bd}^M + 6}, n\right]$. The form of the solution also depends on whether $\alpha\in [1/2, 2/3]$ or $\alpha>2/3$. 

In Table \ref{PL_Optimum2} we describe some optimal primal and dual solutions for the case in which $|V(F)|\in I_1$. It is easy to check that they are feasible and that their common LP value is equal to 
\[\frac{(3+3\ell_{cd}^{M})n-(3+2\ell_{bd}^{M} + \ell_{ad}^{MP})|V(F)|}{3+6\ell_{cd}^{M}}.\]

Note that when $|V(F)|$ is at the right extreme of $I_1$, the LP value above equals $n/2$. Therefore, if $|V(F)|\in I_2$, then the optimal value of LP must be at most $n/2$. In the second column of Table \ref{PL_Optimum3} we show an assignment of value $n/2$, that is dual feasible for $|V(F)|\in I_1\cup I_2$. Using that the primal value is always larger than the value of any dual feasible solution, we conclude that for $|V(F)|\in I_2$, (LP) has value exactly $n/2$, completing the proof of the lemma.

\begin{table}\label{table2}
\centering
\renewcommand{\arraystretch}{2.6}
\setlength{\tabcolsep}{1em}
\begin{tabular}{|c|c|c|}
\hline
Variable &  optimal assignment for $|V(F)|\in I_1$ & feasible assignment for $|V(F)|\in I_1\cup I_2$\\
\hline
$y_1$ &  $\displaystyle \frac{3\ell_{cd}^M - 2\ell_{bd}^M - \ell_{ad}^{MP}}{3(2\ell_{cd}^M +1)}$ & $\displaystyle \frac{1}{2}$ \\
\hline
$y_2$ & $\displaystyle \frac{\ell_{cd}^M - \ell_{bd}^M}{2\ell_{cd}^M +1}$ & $\displaystyle \frac{1}{2}$\\
\hline
$y_3$ & $\displaystyle \frac{\ell_{cd}^M +1 }{2\ell_{cd}^M +1}$ & $\displaystyle \frac{1}{2}$\\
\hline
$y_4$ & $\displaystyle \frac{\ell_{bd}^M - \ell_{ad}^{MP}}{3(2\ell_{cd}^M +1)}$ & $0$\\
\hline
$y_5$ & $\displaystyle \frac{1}{2\ell_{cd}^M +1}$  & $0$\\
\hline
$y_6$ & $\displaystyle \frac{3\ell_{cd}^M - 2\ell_{bd}^M - \ell_{ad}^{MP}}{3(2\ell_{cd}^M +1)}$ & $\displaystyle \frac{1}{2}$\\
\hline
$y_7$ & $0$ & $0$\\
\hline
\end{tabular}
\caption{Optimal solution of the dual problem}\label{PL_Optimum3}
\end{table} 

\begin{proof}[Proof of Lemma \ref{lem:approximation_order_for_local_search_algorithm}]

By Lemma \ref{bound_for_alg}, it holds that $\alg(F)\leq (1-\alpha)(3-f)+3\alpha-2$. 
Using the bound from the LP, Lemma \ref{lema:lpvalue}, that $\Psi_\alpha(\cdot)$ is increasing and that $2\alpha\geq 1$, we get for $f\leq f_\alpha$, \[\frac{\alg(F)}{\opt} \leq (1-\alpha)\Psi_\alpha(f)+\frac{(3\alpha-2)}{\opt}\leq (1-\alpha)\Psi_\alpha(f_{\alpha})+2(3\alpha-2)=2-2(1-\alpha)f_\alpha.\]
On the other hand, if $f\geq f_{\alpha}$, then $\frac{\alg(F)}{\opt} \leq 2\alg(F)\leq2-2(1-\alpha)f_{\alpha}$. \end{proof}

\end{document}